\documentclass[dvipsnames,letterpaper,11pt]{article}
\usepackage[margin=1in]{geometry}
\usepackage{booktabs} 
\usepackage[ruled]{algorithm2e} 
\usepackage[hidelinks]{hyperref}
\usepackage{amsmath,amsthm}
\usepackage{cleveref}
\usepackage{stmaryrd}
\usepackage[longnamesfirst,sort]{natbib}
\usepackage[font={footnotesize}]{caption}
\usepackage{subcaption}
\usepackage{float}
\usepackage{dsfont}
\usepackage{cancel}
\usepackage{xcolor}
\usepackage{soul}
\usepackage{thmtools} 
\usepackage{thm-restate}

\SetAlFnt{\small}
\SetAlCapFnt{\small}
\SetAlCapNameFnt{\small}
\SetAlCapHSkip{0pt}
\IncMargin{-\parindent}
\usepackage{enumitem}
\usepackage{bbold}

\newcommand{\NN}{\mathbb{N}}
\newcommand{\RR}{\mathbb{R}}

\newcommand{\calA}{\mathcal{A}}
\newcommand{\calC}{\mathcal{C}}

\newcommand{\calT}{\mathcal{T}}
\newcommand{\calF}{\mathcal{F}}

\newcommand{\calI}{\mathcal{I}}
\newcommand{\calO}{\mathcal{O}}
\newcommand{\calM}{\mathcal{M}}
\newcommand{\calE}{\mathcal{E}}

\newtheorem{theorem}{Theorem}

\newtheorem{claim}{Claim}

\crefname{claim}{Claim}{Claims}
\crefname{obs}{Observation}{Observation}
\crefname{cons}{constraint}{constraints}
\crefname{eq}{equality}{equalities}
\crefname{eqn}{equation}{equations}
\crefname{ineq}{inequality}{inequalities}

\creflabelformat{cons}{(#2#1#3)}
\creflabelformat{eq}{(#2#1#3)}
\creflabelformat{eqn}{(#2#1#3)}
\creflabelformat{ineq}{(#2#1#3)}

\begin{document}

\title{Near-feasible Fair Allocations in Two-sided Markets\footnote{This work was partially funded by the ANID (Chile) through Grant FONDECYT 1241846.}}

\author{Javier Cembrano\thanks{Department of Algorithms and Complexity, Max Planck Institute for Informatics.}
\and Andr\'es Moraga\thanks{School of Engineering, Pontificia Universidad Cat\'olica de Chile.}
\and Victor Verdugo\thanks{Institute for Mathematical and Computational Engineering, Pontificia Universidad Católica de Chile.} \thanks{Department of Industrial and Systems Engineering, Pontificia Universidad Cat\'olica de Chile.}
}
\date{}
\maketitle

\begin{abstract}
We study resource allocation in two-sided markets from a fundamental perspective and introduce a general modeling and algorithmic framework to effectively incorporate the complex and multidimensional aspects of fairness. 
Our main technical contribution is to show the existence of a range of near-feasible resource allocations parameterized in different model primitives to give flexibility when balancing the different policymaking requirements, allowing policy designers to fix these values according to the specific application.
To construct our near-feasible allocations, we start from a fractional resource allocation and perform an iterative rounding procedure to get an integer allocation.
We show a simple yet flexible and strong sufficient condition for the target feasibility deviations to guarantee that the rounding procedure succeeds, exhibiting the underlying trade-offs between market capacities, agents' demand, and fairness.
To showcase our framework's modeling and algorithmic capabilities, we consider three prominent market design problems: school allocation, stable matching with couples, and political apportionment.
In each of them, we obtain strengthened guarantees on the existence of near-feasible allocations capturing the corresponding fairness notions, such as proportionality, envy-freeness, and stability.
\end{abstract}


\section{Introduction}

Resource allocation is a fundamental task that lies at the intersection of economics, computer science, and operations, where the objective is to develop efficient policies to distribute scarce resources among various agents or entities. 
This naturally results in a two-sided market: One side consists of the agents, while the other side encompasses the different resources that need to be allocated, all while adhering to market capacities and demands.

Many real-world applications involve indivisible resources, a large number of entities, and combinatorial constraints, making the design of efficient resource allocations increasingly complex from both a modeling and computational perspective. 
This two-sided paradigm applies to various situations, including the allocation of students to schools, the composition of political representative bodies, the allocation of doctors to hospitals, organ donation systems, job markets, and online marketplaces, among many others; see, e.g.,~\citet{immorlica2023online,haeringer2018market,roth2018marketplaces,balinski2010fair}.

While efficiency is a key goal in market design, ensuring fairness in resource allocation is arguably one of the most critical factors in assessing the quality of a policy. It is essential to evaluate how effectively the policy treats different groups within a population, ensuring that resources, opportunities, and outcomes are distributed equitably. 
Historical inequalities and biases have often led to disparities among various social, economic, or demographic dimensions, prompting policymakers to enhance policies with proven fairness and efficiency guarantees to rectify these imbalances. 
Typically, these efforts rely on market-specific characteristics to design effective policies, taking advantage of the rich structure of the problem space. However, many existing models and algorithms lack the robustness to address new fairness considerations.

Computational challenges stem from multiple sources.
On one hand, complex combinatorial constraints result in NP-hard problems, requiring a trade-off between efficiency and computational effort. 
Additionally, incorporating fairness requirements across the different dimensions not only adds an extra layer of computational difficulty but also often leads to infeasible scenarios due to deep existing impossibilities in reconciling efficiency with fairness.
To navigate these challenges, implementing {\it near-feasible} solutions, i.e., solutions that may slightly violate some constraints of the problem, is an effective way to overcome these policymaking difficulties.  
From an optimization standpoint, the goal is to provide provable guarantees regarding how close these solutions are to the set of feasible policies.

\subsection{Our Contribution and Techniques} 

In this work, we study resource allocation in two-sided markets from a fundamental perspective and introduce a general modeling and algorithmic framework to effectively incorporate the complex and multidimensional aspects of fairness. 
Within this framework, we provide a rounding theorem to compute near-feasible allocations with strong and flexible approximation guarantees. 
In the following, we summarize our contributions and discuss the implications of our results.

\paragraph{\bf A Multidimensional Resource Allocation Framework.} 
Our first contribution is the introduction of a general framework to model resource allocation problems in two-sided markets.
In this model, there is a set of agents and indivisible resources with finite capacities to be allocated across the agents.
Multiple dimensions characterize the set of agents; each agent may belong to a group within each dimension.
This model feature allows for a multidimensional representation of the agents, e.g., age, gender, ethnicity, and socio-economic or demographic aspects.

The market is provided with target utilities, depending on the dimensions and the groups, that capture the different fairness requirements.
Then, the feasible points of an integer program define the set of feasible resource allocations, though this set may be empty in some cases. 
In our near-feasible resource allocations, the utilities for each group are approximated within a given factor while requiring an extra amount of resources; however, this additional amount is regulated in two ways: by bounding the maximum extra capacity per resource and controlling the total additional market capacity.
In \Cref{sec:prelims}, we provide the formal definition of our multidimensional model, its instances, and our notion of near-feasible resource allocations.

\paragraph{\bf A Rounding Theorem with Flexible Guarantees.} In our multidimensional resource allocation model, the set of feasible resource allocations is represented by the feasible points of an integer program.
From a fundamental point of view, the instances of our multidimensional resource allocation problem can be lifted to get a corresponding hypergraph encoding the agents, resource bundles, dimensions, and groups.
Naturally, strict requirements on the allocation quality might yield infeasible regions for the corresponding integer program.
Our main technical contribution establishes the existence of near-feasible resource allocations, with adjustable approximation guarantees. By tuning model parameters, policymakers can balance and enforce deviation bounds tailored to specific applications.

To construct our near-feasible allocations, we start from a fractional solution of the natural linear relaxation and apply an iterative rounding procedure to derive a feasible solution for the integer program.
We show a simple yet flexible and strong sufficient condition for the target feasibility deviations to guarantee that the iterative rounding procedure computes a near-feasible resource allocation (\Cref{thm:main}).
This condition highlights the underlying trade-offs among market capacities, agents' demand, and utility distribution across dimensions.
In \Cref{sec:main}, we present the formal statement of our main theorem, as well as the rounding algorithm and its analysis. 

\paragraph{\bf Computing Near-feasible Fair Allocations.}
To showcase our framework's modeling and algorithmic capabilities, we consider three prominent market design problems: school allocation, stable matching with couples, and political apportionment.
In each of them, ruled by the sufficient condition of our main rounding theorem, we obtain a range of strengthened guarantees on the existence of near-feasible allocations under general fairness requirements notions, including the proportionality and envy-freeness objectives. 

\paragraph{School allocation.} In \Cref{subsec:group-fairness}, we propose a general multidimensional allocation model which, in particular, captures the school allocation setting introduced by \citet{procacciaetal} as single-dimensional instances.
We model group fairness by a convex optimization-driven approach, and using our rounding \Cref{thm:main} starting from a feasible fractional solution, we get a general guarantee on the existence of near-feasible multidimensional resource allocations (\Cref{cor:group-fairness}).
For the case of proportional fairness, as a corollary, we can directly accommodate the policy-maker priorities by trading off the utility approximation and maximum resource capacity augmentation to get slight constant deviations, which differentiates our result from previous single-dimensional work on group utilities and total allocation excess~\citep{procacciaetal,SMNS24approximation}.

Then, we introduce a new multidimensional notion of group envy-freeness where the ratio between the total utility of a group for its allocation and its total utility for any other group's allocation should not be smaller than the ratio between the sizes of these groups. 
For this natural notion, we provide near-feasible resource allocations with approximation guarantees that do not depend on the number of agents, thus breaking the impossibility found by \citet{procacciaetal} for their more stringent single-dimensional envy-freeness notion (\Cref{prop:envy-free-homog}).

\paragraph{Stable matching with couples.} 
While stable matchings are guaranteed to exist in the basic single-demand setting~\citep{gale1962college}, this is not the case for pairs and, more generally, multi-demand settings. 
To illustrate how our framework can also accommodate stability requirements, we show in \Cref{subsec:couples} how by using our rounding \Cref{thm:main}, we can directly recover a recent guarantee by \citet{nearfeasiblematchingcouples} for the existence of near-feasible stable allocations.
We further extend this setting to handle group fairness requirements and allocation stability to guarantee the existence of near-feasible, stable, and fair allocations (\Cref{prop:couples}).

Our result allows us to incorporate both stability and fairness while keeping the deviations from resource capacities bounded by small constants.
For example, in the single-dimensional case, one can guarantee deviations from group fairness of at most five by increasing the resource capacities by at most four, while keeping the deviation from the total aggregate capacity at two.
Our new near-feasibility results in \Cref{subsec:couples} concern the interaction between stability and fairness, thereby contributing to the efforts on designing two-sided markets in a multidimensional environment under more involved stability concepts. 
We believe our algorithmic framework will help provide near-feasible allocations in further stable matching settings under complex fairness requirements.

\paragraph{Political apportionment.}
In the multidimensional apportionment problem, introduced by \citet{balinskidemange1989a,balinskidemange1989b} for the case of two dimensions and extended by \citet{multidimensionalpoliticalappointment} to an arbitrary number of dimensions, the goal is to allocate the seats of a representative body proportionally across several dimensions. 
Classic apportionment methods, e.g., divisor methods or Hamilton's method, aim to assign seats across districts proportionally to their population or across parties proportionally to their electoral support~\citep{balinski2010fair,pukelsheim2017}.

However, in addition to political parties and geographical divisions, natural dimensions include demographics of the elected members such as gender or ethnicity; see, e.g., \citet{cembranocorreadiazverdugo2021,mathieu2022,bonet2024explainable}.
In \Cref{subsec:apportionment}, using our rounding theorem, we improve over the result by \citet{multidimensionalpoliticalappointment} to get enhanced near-feasibility guarantees for multidimensional apportionment (\Cref{prop:multi-apportionment}).
Remarkably, we can further bound the total deviation from the house size, while the rounding algorithm by \citeauthor{multidimensionalpoliticalappointment} only controls the deviations on each dimension.

\subsection{Further Related Work}

From an algorithmic perspective, our approach to finding near-feasible allocations closely relates to the classic discrepancy minimization problem.
In its basic form, there is a fractional vector $x$ with entries in $[0,1]$ such that $Ax=b$ for some binary matrix $A$ and an integer vector $b$; the goal is to find a binary rounding $y$ minimizing the maximum additive deviation $\|Ax-Ay\|_{\infty}$.

Since the seminal iterative-rounding work by \citet{beck1981integer}, the problem has been extensively studied, including different norms to measure deviations, probabilistic guarantees, general combinatorial constraints, and online variants~\citep{bukh2016improvement,rothvoss2017constructive,bansal2010constructive,bansal2017algorithmic,bansal2019algorithm,lovett2015constructive,bansal2020online}.
Iterative rounding has been a successful tool for designing discrepancy and near-feasibility algorithms, and we refer to the book by \citet{lau2011iterative} for a primer on classic applications.

Very recently, there has been a series of works on near-feasible stable allocations; related to our work is the one by \citet{MGMTSC-groupstabilitycomplexconstr} on group-stability and \citet{OPResearch-stablewithprop} on proportionality, where they provide the existence of near-feasible allocations via rounding methods; in \Cref{sec:discussion}, we briefly discuss the differences between our general fairness approach and the aforementioned works.
In a related line, there are recent works about optimal capacity design for school matching~\citep{capacityvariation}, refuge settlement \citep{delacretaz2016refugee,andersson2020assigning,ahani2021placement}, and healthcare rationing \citep{pathak2021fair,aziz2021efficient}.
In general, the study of fairness in allocation problems has a rich history, and the concepts of proportionality and envy-freeness have been key design objectives; we refer to \citet{moulin2004fair} for an extensive treatment of the fair division theory.

\section{Multidimensional Capacitated Resource Allocation}\label{sec:prelims}

We let $\RR_{+}$ (resp.\ $\RR_{++}$) denote the non-negative (resp.\ strictly positive) reals, $\NN$ denote the strictly positive integers, and $\NN_0= \NN \cup \{0\}$.
We write $[n]$ as a shortcut for $\{1,2,\ldots,n\}$ and $[n]_0$ as a shortcut for $\{0,1,\ldots,n\}$. 
In the {\it multidimensional capacitated resource allocation} problem, or MCRA for short, an instance is structured in the following way:
\begin{enumerate}[label=(\Roman*), leftmargin=15pt]
\item {\bf Agents, resources, and groups.} We have a set $A$ of \textit{entities} or \textit{agents}, a set $R$ of {\it resources}, and a set $A'\subseteq A$ of \textit{binding} agents, i.e., agents that {\it must} be allocated a resource. 
The agents are organized in groups according to $d$ {\it dimensions}, namely, for each dimension $\ell\in [d]$ we have $k_{\ell}$ different groups $G_{\ell,1},\ldots,G_{\ell,k_{\ell}}\subseteq A$ of agents such that for every $i,j\in [k_\ell]$ with $i\neq j$, we have $G_{\ell,i} \cap G_{\ell,j} = \emptyset$, i.e., an agent can be part of at most one group on each dimension.\label{instance:agents}
\item {\bf Demands and capacities.} Each agent $a\in A$ has a {\it demand} $\omega_a\in \NN$ for the number of resources that should receive, and an $\omega_a$-{\it bundle} is any function $q\colon R\to \NN_0$ such that $\sum_{r\in R} q(r)=\omega_a$, where $q(r)$ represents the number of units of resource $r$ that are allocated to agent $a$.
We denote by $\calT_a$ the set of $\omega_a$-bundles, and we write $\omega^*$ for the maximum demand of an agent, i.e., $\max_{a\in A}\omega_a$. 
We finally have a resource \textit{capacity function} $c\colon R\to \NN$, specifying the number of available units of a resource.\label{instance:capacities}
\end{enumerate}
An instance of MCRA is determined by a tuple $\calI=(A,A',R,G,\omega,c)$ ruled by \ref{instance:agents}-\ref{instance:capacities}.
Given an instance of MCRA, we let $\calE = \{(a,q): a\in A, q\in \calT_a\}$ denote the feasible agent-bundle pairs.
A {\it resource allocation} is a mapping $x\colon \calE\to \{0,1\}$ that satisfies the following conditions:
\begin{align}
        \textstyle\sum_{q\in \calT_a} x(a,q) & = 1 \qquad\quad \text{ for every } a\in A',\label[cons]{cons:agents-eq}\\
        \textstyle\sum_{q\in \calT_a} x(a,q) & \leq 1 \qquad\quad \ \text{for every } a\in A \setminus A',\label[cons]{cons:agents-ub}\\
	\textstyle\sum_{a\in A} \sum_{q\in \calT_a} q(r)\cdot x(a,q) & \leq c(r) \quad\quad \text{for every } r\in R.\label[cons]{cons:resources-cap}
\end{align}
Intuitively, we have $x(a,q)=1$ if bundle $q$ is assigned to agent $a$, and $x(a,q)=0$ otherwise.
The set of constraints \eqref{cons:agents-eq} ensures that each agent in $A'$ is allocated exactly one bundle, whereas \eqref{cons:agents-ub} ensures that every other agent is allocated at most one bundle.
The set of constraints \eqref{cons:resources-cap} ensures that at most $c(r)$ units of resource $r$ are allocated. 

We often consider a relaxed notion of resource allocations as a starting point for our algorithms, where we allow bundles to be fractionally allocated.
We say that $x$ is a {\it fractional resource allocation} when it satisfies \eqref{cons:agents-eq}-\eqref{cons:resources-cap} and $x(a,q)\in [0,1]$ for each $(a,q)\in \calE$, i.e., integrality is relaxed.
For a fractional resource allocation $x$, we let $\calA(x)=\{a\in A: |\{q\in \calT_a: x(a,q)>0\}| \geq 2\}$
denote the agents with more than one bundle (fractionally) allocated.

\subsection{Utilities and Near-feasible Allocations}\label{subsec:def-approximation} 
Throughout this work, we consider agent-dependent \textit{utility functions} suitable for our applications. For a collection of utility functions $u_a\colon \calT_a \to \RR_+$ for each $a\in A$, and a (integral or fractional) mapping $x\colon \calE \to[0,1]$, we let
$U_{\ell,i}(x) = \sum_{a\in G_{\ell,i}}\sum_{q\in \calT_a} u_a(q) \cdot x(a,q)$
denote the total utility of group $G_{\ell,i}$, for each $\ell\in [d]$ and $i\in [k_\ell]$.
We further let $U^{*}_{\ell,i} =\max\{u_a(q): a\in G_{\ell,i} \text{ and } q\in \calT_a\}$ denote the maximum utility of an agent in $G_{\ell,i}$.
For notation simplicity, we denote the collection of utility functions $(u_a)_{a\in A}$ by $u$.

\paragraph{\bf Near-feasible allocations.} Let $\alpha\in \NN^d_0$ and $\delta,\Delta\in \NN_0$.
Given an instance $\calI=(A,A',R,G,\omega,c)$ of MCRA, a fractional resource allocation $x$, and a collection of utility functions $u$, a mapping $y\colon \calE\to \{0,1\}$ is an $(\alpha,\delta,\Delta)$-{\it approximation of $x$ with respect to $u$} if it satisfies \eqref{cons:agents-eq}-\eqref{cons:agents-ub} and
\begin{align}
    | U_{\ell,i}(y) - U_{\ell,i}(x)| & <  \alpha_\ell \cdot U^{*}_{\ell,i} \quad \text{for every } \ell\in [d],\ i\in [k_\ell],\label[ineq]{ineq:groups-app}\\
    \textstyle|\sum_{a\in A} \sum_{q\in \calT_a} q(r)\cdot (y(a,q) - x(a,q))| & < \delta \quad\quad\quad\;\; \text{for every } r\in R,\label[ineq]{ineq:resources-app}\\
    \textstyle|\sum_{a\in A} \sum_{q\in \calT_a} \omega_a \cdot (y(a,q)-x(a,q)) | & < \omega^* \cdot \Delta.\label[ineq]{ineq:resources-total-app}
\end{align}
In such near-feasible allocation, the utility of each group $G_{\ell,i}$ deviates strictly less than $\alpha_\ell \cdot U^{*}_{\ell,i}$ from the utilities in $x$, the deviation on the assigned agents to each resource are strictly less than $\delta$, and the deviation on the total number of allocated resources is strictly less than $\omega^*\cdot \Delta$.
As we will illustrate in \Cref{sec:applications}, in some applications the specific structure of the problem guarantees that some of the left-hand sides of these inequalities are integer values, so the maximum possible deviations become $\alpha_\ell \cdot U^{*}_{\ell,i}-1$, $\delta-1$, or $\omega^*\cdot \Delta-1$, respectively.
Finally, we say that a mapping $y\colon \calE\to \{0,1\}$ is a \textit{rounding} of a fractional resource allocation $x$ if $y(a,q)=0$ whenever $x(a,q)=0$ and $y(a,q)=1$ whenever $x(a,q)=1$. 

\section{A Rounding Theorem}\label{sec:main}

In this section, we present our main technical result.
On an intuitive level, the idea is to start from a fractional resource allocation satisfying a certain fairness notion, which is known to be kept upon rounding. 
Our result then states the existence of a resource allocation obtained by rounding this fractional allocation that guarantees small deviations from the resource capacities and from arbitrarily defined utilities.
While having small deviations from capacities constitutes a natural goal in this setting, the definition of the group utilities will allow our model to capture fairness in different applications.
Remarkably, our result gives a high degree of flexibility for choosing the maximum deviations, allowing policy designers to fix these values according to the specific application.

\begin{theorem}\label{thm:main}
    Let $\calI=(A,A',R,G,\omega,c)$ be an instance of MCRA, let $x$ be a fractional resource allocation for this instance, and let $u_a\colon \calT_a \to \RR_+$ be a utility function for each agent $a\in A$. 
    Fix $\psi=1$ if $\calA(x)\neq\emptyset$ or $d\leq 1$, and $\psi\in \{0,1\}$ arbitrarily otherwise.
    Let $\alpha\in \NN^d_0$ and $\delta\in \NN_0$ be such that 
    \begin{equation}
        \frac{\mathds{1}_{\psi=1}}{2} + \sum_{\ell \in [d]}\frac{1}{\alpha_\ell+1} + \frac{\omega^*}{\delta+1} \leq 1\label[ineq]{ineq:deviations}
    \end{equation}
    and let $\Delta \in \NN_0$ be such that one of the following holds:
    \begin{enumerate}[label=\normalfont(\roman*)]
        \item $\Delta\geq 2$ and $\psi=1$; or
        \item \cref{ineq:deviations} holds strictly and $\Delta\geq {1}/\big({1-\big(\frac{\mathds{1}_{\psi=1}}{2} + \sum_{\ell \in [d]} \frac{1}{\alpha_\ell+1} + \frac{\omega^*}{\delta+1} \big)}\big) -1$.
    \end{enumerate}
    Then, there exists a rounding $y$ of $x$ that is an $(\alpha,\delta,\Delta)$-approximation of $x$ with respect to $u$.
    Furthermore, $y$ can be found in time polynomial in $|A|$, $k_1,\ldots,k_d$, and $|R|^{\omega^*}$.
\end{theorem}

The parameter $\psi$ in the theorem captures whether we impose or not, in a certain setting, constraints over the fractional allocation associated with a single agent.
If some agent has two or more associated fractional variables under allocation $x$, i.e., $\calA(x)\neq \emptyset$, this has to be the case to ensure that a fractional allocation is produced; otherwise, we have the freedom to impose these constraints and bound the total deviation $\Delta$ or not impose them and get smaller violations for groups and individual resources.

To prove this theorem, we construct an iterative rounding algorithm that starts from a fractional resource allocation and iteratively rounds its components while ensuring that no constraint is violated by too much.
We first introduce some notation we use for its description and analysis.
For an instance $\calI=(A,A',R,G,\omega,c)$, a subset $\calF\subseteq \calE$, and a mapping $x\colon \calF\to (0,1)$, we write
\begin{align*}
    \calF(a) & = \{(a',q)\in \calF: a'=a\} \quad \qquad \text{for every } a\in A,\\
    \calF(\ell,i) & = \{(a,q)\in \calF: a\in G_{\ell,i}\} \;\;\,\qquad \text{for every } \ell\in [d], i\in [k_\ell],\\
    \calF(r) & = \{(a,q)\in \calF: q(r)\geq 1\} \qquad \;\text{for every } r\in R.
\end{align*}
These concepts allow for a natural hypergraph interpretation, where the vertices are all agents, resources, and groups, and each tuple $(a,q)\in \calF$ is associated with a hyperedge comprising agent $a$, all groups to which $a$ belongs, and all resources in $q$.
With this point of view, the sets $\calF(a)$, $\calF(\ell,i)$, and $\calF(r)$ correspond to the incident hyperedges of a given vertex.
Although this is not a proper hypergraph as a bundle may contain several copies of a resource, we sometimes use this interpretation for informal explanations and discussions.
We often refer to a pair $(a,q)\in \calF$ simply as $e$, in particular, whenever specification about $a$ or $q$ is not needed. 
In slight abuse of notation, we write $u(e)$ instead of $u_a(q)$ when $e=(a,q)$. 
We also write $A(\calF) = \{a\in A: \calF(a)\geq 1\}$, $R(\calF) = \{r\in R: \calF(r)\geq 1\}$, and $G_\ell(\calF) = \{i\in [k_\ell]: \calF(\ell,i)\geq 1\}$ for each $\ell\in [d]$, to refer to the agents, resources, and groups in the $\ell$th dimension with incident hyperedges.

We now introduce the linear program solved by our algorithm in each iteration.
We let $\calI=(A,A',R,G,\omega,c)$ be an instance, $u$ be a collection of utility functions, $\tilde{x}\colon \calF\to (0,1)$ be a mapping on $\calF\subseteq \calE$, $\tilde{A}\subseteq A$ be a subset of binding agents, $\tilde{G}\subseteq \{(\ell,i): \ell\in [d], i\in [k_\ell]\}$ be a subset of binding groups, $\tilde{R}\subseteq R$ be a subset of binding resources, and $\chi\in \{0,1\}$ be a binary value. We consider the following linear program $\mathrm{LP}(\calI,u,\tilde{x},\tilde{A},\tilde{G},\tilde{R},\chi)$ with variables $y\colon \calF\to [0,1]$:
\begin{align}
    \textstyle\sum_{e\in \calF(a)} y(e) & = 1 \qquad\qquad\qquad\qquad\qquad\quad \text{for every } a\in \tilde{A},\label[cons]{cons:bind-agents-no-dev}\\
    \textstyle\sum_{e\in \calF(a)} y(e) & \leq 1  \qquad\qquad\qquad\qquad\qquad\;\;\; \text{for every } a\in A(\calF)\setminus \tilde{A},\label[cons]{cons:agents-no-dev}\\
    \textstyle\sum_{e\in \calF(\ell,i)} u(e)\cdot y(e) & = \textstyle\sum_{e\in \calF(\ell,i)} u(e)\cdot \tilde{x}(e) \qquad\quad\  \text{for every } \ell\in [d],\ i\in \tilde{G}_\ell,\label[cons]{cons:groups-no-dev}\\
    \textstyle\sum_{(a,q)\in \calF(r)} q(r) \cdot y(a,q) & = \textstyle\sum_{(a,q)\in \calF(r)} q(r) \cdot \tilde{x}(a,q)\quad\ \;\text{for every } r\in \tilde{R},\label[cons]{cons:resources-no-dev}\\
    \textstyle\sum_{(a,q)\in \calF} \omega_a \cdot y(a,q) & = \textstyle\sum_{(a,q)\in \calF} \omega_a \cdot \tilde{x}(a,q) \quad\qquad\; \text{if } \chi=1,\label[cons]{cons:resources-total-no-dev}\\
    \textstyle0\leq y(e)& \leq 1 \qquad\qquad\qquad\qquad\qquad\quad \text{for every } e\in \calF.\label[cons]{cons:domain-no-dev}
\end{align}
In each iteration of our algorithm, $\calF$ will represent the remaining tuples $(a,q)$ with an associated fractional value $\tilde{x}(a,q)$ in the allocation. 
Solving this linear program will then ensure to maintain a total allocation of $1$ for all agents in $\tilde{A}$ and at most one for the other agents, the same total utility as in $\tilde{x}$ for every group $G_{\ell,i}$ with $i\in \tilde{G}_\ell$, the same number of allocated agents for every resource in $\tilde{R}$, and the same total number of allocated agents when $\chi=1$.
The definitions of these sets of agents, groups, and resources in each iteration will be made to balance two simultaneous objectives: (i) that agents, groups, and resources with many associated fractional values do not deviate from their current allocation, and (ii) that there are more variables than linearly independent equality constraints so that the algorithm makes progress and eventually terminates. Similarly, the binary value $\chi$ will be set to zero only when there are few fractional values, to ensure that the total deviation from the number of resources allocated initially is kept under control.

Our rounding algorithm, formally described in \Cref{alg:main}, starts from a given fractional allocation $x$ (which we also call $x^0$) and iteratively solves $\mathrm{LP}(\calI,u,\tilde{x},\tilde{A},\tilde{G},\tilde{R},\chi)$ for $\tilde{x}$ being the fractional entries of $x$, each time only considering the necessary constraints to prevent deviations beyond the allowed ones.
At step $t$, the algorithm starts from an allocation $x^t$ and constructs an allocation $x^{t+1}$ by leaving the integral entries of $x^t$ unchanged and taking, for the fractional entries $\tilde{x}^t$, an extreme point of this linear program.

This program considers, for instance, \cref{cons:groups-no-dev} for each group $G_{\ell,i}$ such that there are $\alpha_\ell+1$ or more fractional values $\tilde{x}^{t}(a,q)$ for some $a\in G_{\ell,i}$ and $q\in \calT_a$; as the other groups $G_{\ell,i}$ have at most $\alpha_\ell$ such fractional values, they will deviate less than $\alpha_\ell \cdot U^*_{\ell,i}$ from their previous utility.
Similarly, the program considers \cref{cons:bind-agents-no-dev} for all agents with two or more values of $\tilde{x}^{t}(a,q)$ summing up to $1$, it considers \cref{cons:resources-no-dev} for all resources $r\in R$ with $\delta+1$ or more units associated with fractional values of $\tilde{x}^{t}$, and it considers \cref{cons:resources-total-no-dev} only if there are $\Delta+1$ or more tuples $(a,q)$ with agent $a$ not being considered for \cref{cons:bind-agents-no-dev}.
We will prove that this construction ensures having more variables than linearly independent equality constraints in each iteration.
Thus, during the algorithm execution, we either fix a new variable to an integer value or an inequality constraint becomes tight.
\begin{algorithm}[t]
\caption{Iterative rounding for MCRA}
\label{alg:main}    
\SetAlgoNoLine
\KwIn{instance $\calI=(A,A',R,G,\omega,c)$, utility functions $u$, fractional resource allocation $x$, maximum deviations $\alpha\in \NN^d_0$ and $\delta,\Delta\in \NN_0$}
\KwOut{resource allocation $y$}
$x^0\gets x$\;
$t\gets 0$\;
\While{\bf true}{
    $\calF^t\gets \{(a,q)\in \calE: 0<x^t(a,q)<1 \}$\;
    $\tilde{x}^t(e)\gets x^t(e)$ for every $e\in \calF^t$\;
    $\tilde{A}^t \gets \big\{a\in \calA(x^t): \sum_{e\in {\calF^t}(a)} x^t(e) = 1\big\}$\;
    $\tilde{G}^t_\ell \gets \big\{(\ell,i): i\in [k_\ell] \text{ s.t.\ } |{\calF^t}(\ell,i)| \geq \alpha_\ell+1\big\}$ for every $\ell\in [d]$\;
    $\tilde{R}^t \gets \big\{r\in R: \sum_{(a,q)\in {\calF^t}(r)} q(r) \geq \delta+1\big\}$\;
    $\chi^t \gets 1$ {\bf if } $|A(\calF^t)\setminus \tilde{A}^t| \geq \Delta+1$, {\bf else} $\chi^t \gets 0$\;
    \If{$\tilde{A}^t = \tilde{R}^t = \emptyset$, $\tilde{G}^t_\ell=\emptyset$ for every $\ell\in [d]$, and $\chi^t=0$}{
        {\bf break}
    }
    $x^{t+1}(e)\gets x^t(e)$ for every $e\in \calE \setminus \calF^t$\;
    $y^* \gets$ extreme point of $\mathrm{LP}(\calI,u,\tilde{x}^t,\tilde{A}^t,\tilde{G}^t,\tilde{R}^t,\chi^t)$\;
    $x^{t+1}(e)\gets y^*(e)$ for every $e\in \calF^t$\;
    $t\gets t+1$
}
$T\gets t$\;
$x^{T+1}(e)\gets x^T(e)$ for every $e\in \calE \setminus \calF^T$\;
fix $x^{T+1}(e) \in \{\lfloor x^{T}(e)\rfloor, \lceil x^{T}(e)\rceil\}$ for every $e\in \calF^T$ such that $\sum_{e\in {\calF^T}(a)} x^{T+1}(e) \leq 1$ for each $a\in A(\calF^T)$\;
{\bf return} $x^{T+1}$
\end{algorithm}

When analyzing the algorithm, we usually fix its input and directly refer to the objects constructed during its execution with the names defined therein.
For an iteration $t\in [T-1]_0$, we further let
\[
    \textstyle \calC^t = |\tilde{A}^t| + \sum_{\ell\in [d]}|\tilde{G}^t_\ell| + |\tilde{R}^t| + \chi^t
\]
denote the number of equality constraints of $\mathrm{LP}(\calI,u,\tilde{x}^t,\tilde{A}^t,\tilde{G}^t,\tilde{R}^t,\chi^t)$.
The following lemma states the property that ensures that the algorithm terminates: The linear program solved in each iteration has more variables than linearly independent equality constraints. 
\begin{restatable}{lemma}{lemAlgorithmTerminates}\label{lem:algorithm-terminates}
    Let $\calI=(A,A',R,G,\omega,c)$ be an instance, $x$ a fractional resource allocation for it, $u$ a collection of utility functions, and $\alpha,\delta,\Delta,\psi$ such that the conditions in the statement of \Cref{thm:main} are satisfied. 
    Then, for every iteration $t\in [T-1]_0$ it holds that $\calC^t \leq |\calF^t|$.
    Furthermore, if the inequality is tight, then the following properties hold:
    \begin{enumerate}[label=\normalfont(\roman*)]
        \item If $\psi=1$, then $\tilde{A}^t=A(\calF^t)$;\label{lem:algorithm-terminates-i}
        \item $\tilde{G}^t_\ell=G_\ell(\calF^t)$ and $\bigcup_{i\in G_\ell(\calF^t)} G_{\ell,i} = A(\calF^t)$ for every $\ell\in [d]$;\label{lem:algorithm-terminates-ii}
        \item $\tilde{R}^t=R(\calF^t)$ and $\sum_{r\in R(\calF^t)}q(r)=\omega^*$ 
        for every $(a,q)\in \calF^t$.\label{lem:algorithm-terminates-iii}
    \end{enumerate}
\end{restatable}

Before proving this result, we state a simple lemma that covers the case where $\Delta$ satisfies the second property in the statement of \Cref{thm:main}.

\begin{restatable}{lemma}{lemDeltaCaseTwo}\label{lem:Delta-case2}
    Let $d\in \NN$ and $\Theta,\gamma\in \RR^d_{++}$ be such that $\sum_{\ell\in [d]} \frac{\Theta_\ell}{\gamma_\ell} < 1$. Then, for every $\theta,\varepsilon \in \RR^d_+$ with $\theta_\ell\leq \Theta_\ell$ for every $\ell\in [d]$ and every $z\in \NN$ we have that
    \[
         \sum_{\ell\in [d]}\bigg\lfloor\frac{\theta_\ell z - \varepsilon_\ell}{\gamma_\ell} \bigg\rfloor \leq z-1.
    \]
    Furthermore, if the inequality is tight, we have either 
    \begin{enumerate}[label=\normalfont(\roman*)]
        \item $\varepsilon_\ell=0$ and $\theta_\ell=\Theta_\ell$ for every $\ell\in [d]$; or
        \item $z< {1}/{(1-\sum_{\ell\in [d]}\frac{\Theta_\ell}{\gamma_\ell})}$.
    \end{enumerate}
\end{restatable}

\begin{proof}
Let $d,\Theta,\gamma$ be as in the statement, and let first $\theta,\varepsilon\in \RR^d_+$ and $z\in \NN$ be arbitrary numbers such that $\theta_\ell\leq \Theta_\ell$ for every $\ell\in [d]$.
Then,
\[
    \sum_{\ell\in [d]}\bigg\lfloor\frac{\theta_\ell z - \varepsilon_\ell}{\gamma_\ell} \bigg\rfloor \leq \sum_{\ell\in [d]}\frac{\Theta_\ell z}{\gamma_\ell} < z.
\]
Since the expression on the left-hand side is an integer, the upper bound of $z-1$ follows.

Let now $\theta,\varepsilon\in \RR^d_+$ and $z\in \NN$ be such that $\theta_\ell\leq \Theta_\ell$ and
\[
    \sum_{\ell\in [d]}\bigg\lfloor\frac{\theta_\ell z - \varepsilon_\ell}{\gamma_\ell} \bigg\rfloor = z-1.
\]
Suppose that there exists $\ell'\in [d]$ such that either $\varepsilon_{\ell'}>0$ or $\theta_{\ell'}<\Theta_{\ell'}$.
We define $r_\ell = \frac{\Theta_\ell z }{\gamma_\ell} - \big\lfloor \frac{\theta_\ell z -\varepsilon_\ell}{\gamma_\ell} \big\rfloor$ for each $\ell\in [d]$. 
Clearly, $r_\ell \geq 0$ for all $\ell \in [d]$ and $r_{\ell'} > 0$.
Taking the sum of $r_\ell$ over all $\ell \in [d]$, we then get
\[
    0 < \sum_{\ell \in [d]} r_\ell = \sum_{\ell \in [d]} \frac{\Theta_\ell z }{\gamma_\ell} - \sum_{\ell \in [d]} \bigg\lfloor \frac{\theta_\ell z-\varepsilon_\ell}{\gamma_\ell} \bigg\rfloor = \sum_{\ell \in [d]}\frac{\Theta_\ell z }{\gamma_\ell} - (z-1),
\]
where the last equality comes from the previous assumption.
Rearranging the extreme terms, we obtain $z < {1}/{(1-\sum_{\ell\in [d]}\frac{\Theta_\ell}{\gamma_\ell})}$.
\end{proof}

We now proceed with the proof of \Cref{lem:algorithm-terminates}.

\begin{proof}[Proof of \Cref{lem:algorithm-terminates}]
Let $\calI=(A,A',R,G,\omega,c)$, $u$, $\alpha$, $\delta$, $\Delta$, and $\psi$ be as in the statement.
Let also $t\in [T-1]_0$ be any fixed step of \Cref{alg:main}. Note that the lemma only applies to some $t$ if $T\geq 1$, i.e., if the algorithm solves the linear program at least once.

From the fact that \Cref{alg:main} fixes $x^{t'+1}(e)=x^{t'}(e)$ for every $e$ such that $x^{t'}(e)\in \{0,1\}$ and every $t'\in [t-1]_0$, we know that $\calA(x^{t'+1})\subseteq \calA(x^{t'})$ for every $t'\in [t-1]_0$.
Since $\tilde{A}^t\subseteq \calA(x^t)$ by definition, we conclude that $\tilde{A}^t\subseteq \calA(x^0)=\calA(x)$.
In particular, we will make use of the fact that $\tilde{A}^t=\emptyset$ whenever $\psi=0$.
Furthermore, we know that ${\calF^t}(a) \geq 2$ for every $a \in \tilde{A}^t$. Summing this inequality over $a\in \tilde{A}^t$ and rearranging, we obtain
\begin{align}
    |\tilde{A}^t| & \leq \bigg\lfloor \frac{\sum_{a\in \tilde{A}^t} {\calF^t}(a)}{2} \bigg\rfloor \mathds{1}_{\psi=1} \nonumber\\
    & = \bigg\lfloor \frac{\sum_{a\in A(\calF^t)} {\calF^t}(a) - \sum_{a\in A(\calF^t) \setminus \tilde{A}^t} {\calF^t}(a)}{2} \bigg\rfloor \mathds{1}_{\psi=1} \nonumber\\
    & \leq \bigg\lfloor \frac{|\calF^t| - |A(\calF^t)\setminus \tilde{A}^t|}{2}\bigg\rfloor \mathds{1}_{\psi=1},\label[ineq]{ineq:bd-active-agents}
\end{align}
where the last inequality follows from the definitions of ${\calF^t}(a)$ and $A(\calF^t)$.

Similarly, for each $\ell\in [d]$ we know from the definition of $\tilde{G}^t_\ell$ that ${\calF^t}(\ell,i) \geq \alpha_\ell+1$ for every $i\in \tilde{G}^t_\ell$. Fixing $\ell\in [d]$, summing the previous inequality over $i\in \tilde{G}^t_\ell$, and rearranging, we obtain
\begin{align}
    |\tilde{G}^t_\ell| & \leq \bigg\lfloor \frac{\sum_{i\in \tilde{G}^t_\ell} {\calF^t}(\ell,i)}{\alpha_\ell+1} \bigg\rfloor\nonumber\\
    & = \bigg\lfloor \frac{\big|\bigcup_{i\in \tilde{G}^t_\ell} {\calF^t}(\ell,i)\big|}{\alpha_\ell+1} \bigg\rfloor\nonumber\\
    & = \bigg\lfloor \frac{|\calF^t| - \big|\calF^t \setminus \bigcup_{i\in G_\ell(\calF^t)} {\calF^t}(\ell,i)\big| - \big|\bigcup_{i\in G_\ell(\calF^t)} {\calF^t}(\ell,i) \setminus \bigcup_{i\in \tilde{G}^t_\ell} {\calF^t}(\ell,i)\big|}{\alpha_\ell+1} \bigg\rfloor\label[ineq]{ineq:bd-active-groups}
\end{align}
where the first equality uses that $|\{i\in \tilde{G}^t_\ell: e\in {\calF^t}(\ell,i)\}|\leq 1$ for every $e\in \calF^t$ and the last one follows from the fact that
\[
    \bigcup_{i\in \tilde{G}^t_\ell} {\calF^t}(\ell,i) \subseteq \bigcup_{i\in G_\ell(\calF^t)} {\calF^t}(\ell,i) \subseteq \calF^t.
\]

Finally, from the definition of $\tilde{R}^t$ we know that $\sum_{(a,q)\in {\calF^t}(r)}q(r) \geq \delta+1$ for every $r \in \tilde{R}^t$. Summing over $r\in \tilde{R}^t$ and rearranging, we obtain
\begin{align}
    |\tilde{R}^t| & \leq \bigg\lfloor \frac{\sum_{r\in \tilde{R}^t} \sum_{(a,q)\in {\calF^t}(r)}q(r)}{\delta+1} \bigg\rfloor \nonumber\\
    & = \bigg\lfloor \frac{\sum_{(a,q)\in \calF^t}\sum_{r\in R(\calF^t)} q(r) - \sum_{r\in R(\calF^t)\setminus \tilde{R}^t}\sum_{(a,q)\in {\calF^t}(r)} q(r)}{\delta+1} \bigg\rfloor \nonumber\\
    & \leq \bigg\lfloor \frac{\sum_{(a,q)\in \calF^t} \sum_{r\in R(\calF^t)} q(r) -|R(\calF^t)\setminus \tilde{R}^t|}{\delta+1} \bigg\rfloor,\label[ineq]{ineq:bd-active-resources}
\end{align}
where the last inequality follows from the definitions of ${\calF^t}(r)$ and $R(\calF^t)$.

We now use \cref{ineq:bd-active-agents,ineq:bd-active-groups,ineq:bd-active-resources} to conclude the result.
We first consider the case where $\chi^t=0$.
In this case, we claim that
\begin{align*}
    \calC^t & \leq \bigg\lfloor \frac{|\calF^t| - |A(\calF^t)\setminus \tilde{A}^t|}{2}\bigg\rfloor \mathds{1}_{\psi=1} + \bigg\lfloor \frac{\sum_{(a,q)\in \calF^t} \sum_{r\in R(\calF^t)} q(r) -|R(\calF^t)\setminus \tilde{R}^t|}{\delta+1} \bigg\rfloor \\
    & \phantom{{}\leq{}} + \sum_{\ell\in [d]} \bigg\lfloor \frac{|\calF^t| - \big|\calF^t \setminus \bigcup_{i\in G_\ell(\calF^t)} {\calF^t}(\ell,i)\big| - \big|\bigcup_{i\in G_\ell(\calF^t)} {\calF^t}(\ell,i) \setminus \bigcup_{i\in \tilde{G}^t_\ell} {\calF^t}(\ell,i)\big|}{\alpha_\ell+1} \bigg\rfloor \\
    & \leq \bigg\lfloor \frac{|\calF^t|}{2}\bigg\rfloor\mathds{1}_{\psi=1} + \bigg\lfloor \frac{\omega^*|\calF^t|}{\delta+1} \bigg\rfloor + \sum_{\ell\in [d]} \bigg\lfloor \frac{|\calF^t|}{\alpha_\ell+1}\bigg\rfloor \\ 
    & \leq \frac{|\calF^t|}{2} \mathds{1}_{\psi=1} + \frac{\omega^*|\calF^t|}{\delta+1} + \sum_{\ell\in [d]} \frac{|\calF^t|}{\alpha_\ell+1} = \bigg( \frac{\mathds{1}_{\psi=1}}{2} + \frac{\omega^*}{\delta+1} + \sum_{\ell \in [d]}\frac{1}{\alpha_\ell+1} \bigg) |\calF^t| \leq |\calF^t|.
\end{align*}
Indeed, the first inequality follows from \cref{ineq:bd-active-agents,ineq:bd-active-groups,ineq:bd-active-resources}, the second one uses the fact that $\sum_{r\in R(\calF^t)} q(r) \leq \omega^*$ for every $q\in \calT_a$ and every $a\in A$, and the last one comes from the fact that $\alpha$ and $\delta$ satisfy the conditions in the statement of \Cref{thm:main}.
The third inequality and the equality follow from simple calculations.
This allows us to conclude the inequality in the statement.
Moreover, in order to have equality throughout, we need the second inequality to be tight, which yields
\begin{enumerate}[label=(\roman*)]
    \item $|A(\calF^t)\setminus \tilde{A}^t| = 0$ if $\psi=1$;
    \item $\big|\calF^t \setminus \bigcup_{i\in G_\ell(\calF^t)} {\calF^t}(\ell,i)\big| = \big|\bigcup_{i\in G_\ell(\calF^t)} {\calF^t}(\ell,i) \setminus \bigcup_{i\in \tilde{G}^t_\ell} {\calF^t}(\ell,i)\big| = 0$ for every $\ell\in [d]$;
    \item $|R(\calF^t)\setminus \tilde{R}^t| = 0$ and $\sum_{r\in R(\calF^t)} q(r) = \omega^*$ for every $(a,q)\in \calF^t$.
\end{enumerate}
It is not hard to see that these properties are equivalent to \cref{lem:algorithm-terminates-i,lem:algorithm-terminates-ii,lem:algorithm-terminates-iii} in the statement.

We now consider the case with $\chi^t=1$.
From the statement of \Cref{thm:main}, we have two subcases: either \cref{ineq:deviations} holds strictly and $\Delta\geq \frac{1}{1-\big(\frac{\mathds{1}_{\psi=1}}{2} + \sum_{\ell \in [d]} \frac{1}{\alpha_\ell+1} + \frac{\omega^*}{\delta+1}\big)} -1$, or $\Delta\geq 2$ and $\psi=1$.

In the first case, we claim that
\begin{align*}
    \calC^t & \leq \bigg\lfloor \frac{|\calF^t| - |A(\calF^t)\setminus \tilde{A}^t|}{2}\bigg\rfloor \mathds{1}_{\psi=1} + \Bigg\lfloor \frac{\frac{\sum_{(a,q)\in \calF^t} \sum_{r\in R(\calF^t)} q(r)}{|\calF^t|}|\calF^t|-|R(\calF^t)\setminus \tilde{R}^t|}{\delta+1} \Bigg\rfloor \\
    & \phantom{{}\leq{}} + \sum_{\ell\in [d]} \bigg\lfloor \frac{|\calF^t| - \big|\calF^t \setminus \bigcup_{i\in G_\ell(\calF^t)} {\calF^t}(\ell,i)\big| - \big|\bigcup_{i\in G_\ell(\calF^t)} {\calF^t}(\ell,i) \setminus \bigcup_{i\in \tilde{G}^t_\ell} {\calF^t}(\ell,i)\big|}{\alpha_\ell+1} \bigg\rfloor + 1\\
    & \leq |\calF^t|.
\end{align*}
Indeed, the first inequality follows from \cref{ineq:bd-active-agents,ineq:bd-active-groups,ineq:bd-active-resources} and the second one from \Cref{lem:Delta-case2}.
The lemma applies directly because 
\begin{align}
    \frac{\mathds{1}_{\psi=1}}{2} + \sum_{\ell\in [d]}\frac{1}{\alpha_\ell+1} + \frac{\omega^*}{\delta+1}  & < 1,\nonumber \\
    \frac{\sum_{(a,q)\in \calF^t} \sum_{r\in R(\calF^t)} q(r)}{|\calF^t|} &\leq \omega^*, \text{ and} \label[ineq]{ineq:avg-multiplicity}\\
    |\calF^t| \geq |A(\calF^t)\setminus \tilde{A}^t| \geq \Delta+1 & \geq \frac{1}{1-\big(\frac{\mathds{1}_{\psi=1}}{2} + \sum_{\ell \in [d]} \frac{1}{\alpha_\ell+1} + \frac{\omega^*}{\delta+1} \big)},\nonumber
\end{align}
where the first inequality follows from the assumption that \cref{ineq:deviations} holds strictly, the second one from the fact that $\sum_{r\in R(\calF^t)} q(r) \leq \omega^*$ for every $q\in \calT_a$ and $a\in A$, and the last one from the definition of $\chi^t$ in \Cref{alg:main} and the lower bound that we assumed on $\Delta$.
Furthermore, if we have $\calC^t = |\calF^t|$, \Cref{lem:Delta-case2} implies that
\begin{enumerate}[label=(\roman*)]
    \item $|A(\calF^t)\setminus \tilde{A}^t| = 0$ if $\psi=1$;
    \item $\big|\calF^t \setminus \bigcup_{i\in G_\ell(\calF^t)} {\calF^t}(\ell,i)\big| = \big|\bigcup_{i\in G_\ell(\calF^t)} {\calF^t}(\ell,i) \setminus \bigcup_{i\in \tilde{G}^t_\ell} {\calF^t}(\ell,i)\big| = 0$ for every $\ell\in [d]$;
    \item $|R(\calF^t)\setminus \tilde{R}^t| = 0$ and \cref{ineq:avg-multiplicity} is tight, thus $\sum_{r\in R(\calF^t)} q(r)=\omega^*$ for every $(a,q)\in \calF^t$.
\end{enumerate}
As before, these properties are equivalent to \cref{lem:algorithm-terminates-i,lem:algorithm-terminates-ii,lem:algorithm-terminates-iii} in the statement.

We finally consider the case with $\chi^t=1$, $\Delta\geq 2$ and $\psi=1$.
From the definition of $\chi^t$ in \Cref{alg:main}, we know that
\begin{equation}
    |A(\calF^t)\setminus \tilde{A}^t| \geq \Delta+1 \geq 3,\label[ineq]{ineq:bound-Delta-chi1-case1}
\end{equation}
and from \cref{ineq:deviations} in \Cref{thm:main} we know that
\begin{equation}
    \sum_{\ell \in [d]}\frac{1}{\alpha_\ell+1} + \frac{\omega^*}{\delta+1} \leq \frac{1}{2}.\label[ineq]{ineq:sum_devs-chi1}
\end{equation}
Therefore, we obtain
\begin{align*}
    \calC^t & \leq \bigg\lfloor \frac{|\calF^t| - |A(\calF^t)\setminus \tilde{A}^t|}{2}\bigg\rfloor + \sum_{\ell\in [d]} \bigg\lfloor \frac{|\calF^t|}{\alpha_\ell+1}\bigg\rfloor + \bigg\lfloor \frac{\omega^*|\calF^t|}{\delta+1} \bigg\rfloor + 1\\
    & \leq \bigg\lfloor \frac{|\calF^t| - 3}{2}\bigg\rfloor + \sum_{\ell\in [d]} \bigg\lfloor \frac{|\calF^t|}{\alpha_\ell+1}\bigg\rfloor + \bigg\lfloor \frac{\omega^*|\calF^t|}{\delta+1} \bigg\rfloor + 1\\
    & \leq \frac{|\calF^t|}{2} - \frac{3}{2} + \sum_{\ell\in [d]} \frac{|\calF^t|}{\alpha_\ell+1} + \frac{\omega^*|\calF^t|}{\delta+1} + 1\\
    & = \bigg( \frac{1}{2} + \sum_{\ell \in [d]}\frac{1}{\alpha_\ell+1} + \frac{\omega^*}{\delta+1} \bigg) |\calF^t| - \frac{1}{2} \leq |\calF^t| - \frac{1}{2},
\end{align*}
where the first inequality follows from \cref{ineq:bd-active-agents,ineq:bd-active-groups,ineq:bd-active-resources} and the facts that $\psi=1$ and $\sum_{r\in R(\calF^t)} q(r) \leq \omega^*$ for every $q\in \calT_a$ and $a\in A$, the second one from \cref{ineq:bound-Delta-chi1-case1}, and the last one from \cref{ineq:sum_devs-chi1}.
The third inequality and the equality follow from simple calculations.
Since $\calC^t$ and $|\calF^t|$ are integers, we conclude that $\calC^t\leq |\calF^t|-1$.
\end{proof}

Equipped with \Cref{lem:algorithm-terminates}, we now show \Cref{thm:main}. 
Given an instance and a fractional allocation $x$, we prove that the outcome $y$ of \Cref{alg:main} satisfies all conditions stated in the theorem. 
That the algorithm terminates follows from \Cref{lem:algorithm-terminates}, as we show that the properties stated therein when the inequality is tight contradict the linear independence of the set of equality constraints. 
The running time is obtained by bounding the size of the linear program solved in each step and observing that its size also constitutes an upper bound for the number of iterations of the algorithm.
That $y$ satisfies the claimed notion of approximation is the most demanding part of the proof, requiring an understanding of the number and structure of the fractional variables upon termination of the iterative rounding procedure.

\begin{proof}[Proof of \Cref{thm:main}]
We consider an instance $\calI=(A,A',R,G,\omega,c)$, a fractional allocation $x$, utility functions $u_a\colon \calT_a \to \RR_+$ for each $a\in A$, and values $\psi\in\{0,1\}$, $\alpha\in \NN^d_0$, $\delta\in \NN_0$, and $\Delta \in \NN_0$ satisfying the conditions in the statement.
We let $y=x^{T+1}$ be the outcome of \Cref{alg:main} with this input, and we claim the result for $y$.

We first argue that \Cref{alg:main} indeed terminates and produces an outcome $y$.
To see this, let $t\in [T-1]_0$ be an arbitrary step and observe that $\mathrm{LP}(\calI,u,\tilde{x}^t,\tilde{A}^t,\tilde{G}^t,\tilde{R}^t,\chi^t)$ has $|\calF^t|$ variables and $\calC^t$ equality constraints.
From \Cref{lem:algorithm-terminates}, we thus know that the linear program either has more variables than equality constraints, or it has the same number of variables as equality constraints and \cref{lem:algorithm-terminates-i,lem:algorithm-terminates-ii,lem:algorithm-terminates-iii} in the statement of that lemma hold. 
We claim that, in the latter case, the set of equality constraints is not linearly independent.
\begin{restatable}{claim}{claimConstraintsld}\label{claim:constraints-ld}
    If $\calC^t=|\calF^t|$, the equality constraints of $\mathrm{LP}(\calI,u,\tilde{x}^t,\tilde{A}^t,\tilde{G}^t,\tilde{R}^t,\chi^t)$ are not linearly independent.
\end{restatable}
\begin{proof}
Adding constraints \eqref{cons:resources-no-dev} over all resources $r\in \tilde{R}^t$ yields
\allowdisplaybreaks
\begin{align}
    & \sum_{r\in \tilde{R}^t} \sum_{(a,q)\in {\calF^t}(r)} q(r) \cdot y(a,q) = \sum_{r\in \tilde{R}^t}\sum_{(a,q)\in {\calF^t}(r)} q(r) \cdot \tilde{x}^t(a,q)\nonumber\\
    \Longleftrightarrow & \sum_{r\in R(\calF^t)} \sum_{(a,q)\in {\calF^t}(r)} q(r) \cdot y(a,q) = \sum_{r\in R(\calF^t)}\sum_{(a,q)\in {\calF^t}(r)} q(r) \cdot \tilde{x}^t(a,q)\nonumber\\
    \Longleftrightarrow & \sum_{(a,q)\in \calF^t}\sum_{r\in R(\calF^t)} q(r) \cdot y(a,q) = \sum_{(a,q)\in \calF^t}\sum_{r\in R(\calF^t)} q(r) \cdot \tilde{x}^t(a,q) \nonumber\\
    \Longleftrightarrow & \sum_{(a,q)\in \calF^t}\omega^* \cdot y(a,q) = \sum_{(a,q)\in \calF^t}\omega^* \cdot \tilde{x}^t(a,q)\nonumber\\
    \Longleftrightarrow & \sum_{(a,q)\in \calF^t} y(a,q) = \sum_{(a,q)\in \calF^t} \tilde{x}^t(a,q),\label{cons:resources-dep}
\end{align}
where we used \cref{lem:algorithm-terminates-iii} of \Cref{lem:algorithm-terminates} to replace $\tilde{R}^t$ by $R(\calF^t)$ in the first equivalence and $\sum_{r\in R(\calF^t)} q(r)$ by $\omega^*$ in the third equivalence.

We now make use of the fact that, due to the definition of $\psi$, we either have $\psi=1$ or $d\geq 2$ (or both). If $\psi=1$, adding constraints \eqref{cons:bind-agents-no-dev} over all agents $a\in \tilde{A}^t$ yields
\[
    \sum_{a\in \tilde{A}^t} \sum_{e\in {\calF^t}(a)} y(e) = |\tilde{A}^t| \Longleftrightarrow \sum_{a\in A(\calF^t)} \sum_{e\in {\calF^t}(a)} y(e) = \sum_{e\in \calF^t} \tilde{x}^t(e) \Longleftrightarrow \sum_{e\in \calF^t} y(e) = \sum_{e\in \calF^t} \tilde{x}^t(e),
\]
where we used \cref{lem:algorithm-terminates-i} of \Cref{lem:algorithm-terminates} to replace $\tilde{A}^t$ by $A(\calF^t)$ and the fact that $\sum_{e\in \calF^t(a)} \tilde{x}^t(e) = 1$ for every $a\in \tilde{A}^t$.
Since this equality and equality \eqref{cons:resources-dep} are the same, we conclude that the set of constraints is not linearly independent, as claimed.
If $d\geq 2$, adding constraints \eqref{cons:groups-no-dev} over all $i\in \tilde{G}^t_\ell$ for any fixed $\ell\in [d]$ yields
\begin{align*}
    & \sum_{i\in \tilde{G}^t_\ell} \sum_{e\in {\calF^t}(\ell,i)} u(e) \cdot y(e) = \sum_{i\in \tilde{G}^t_\ell} \sum_{e\in {\calF^t}(\ell,i)} u(e) \cdot \tilde{x}^t(e) \\
    \Longleftrightarrow & \sum_{i\in G_\ell(\calF^t)} \sum_{e\in {\calF^t}(\ell,i)} u(e) \cdot y(e) = \sum_{i\in G_\ell(\calF^t)} \sum_{e\in {\calF^t}(\ell,i)} u(e) \cdot \tilde{x}^t(e) \\
    \Longleftrightarrow & \sum_{e\in \calF^t} u(e)\cdot y(e) = \sum_{e\in \calF^t} u(e)\cdot \tilde{x}^t(e),
\end{align*}
where we used \cref{lem:algorithm-terminates-ii} of \Cref{lem:algorithm-terminates} to replace $\tilde{G}^t_\ell$ by $G_\ell(\calF^t)$ in the first equivalence and to conclude that $\cup_{i\in G_\ell(\calF^t)} {\calF^t}(\ell,i) = \cup_{a\in A} {\calF^t}(a) = \calF^t$ due to $\bigcup_{i\in G_\ell(\calF^t)}G_{\ell,i}=A(\calF^t)$ in the last equivalence.
Since the last equality is independent of $\ell$, whenever $d\geq 2$ we have that the constraints for $\ell,\ell'\in [d]$ with $\ell\neq \ell'$ are linearly dependent.
\end{proof}

Due to \Cref{claim:constraints-ld}, whenever $\mathrm{LP}(\calI,u,\tilde{x}^t,\tilde{A}^t,\tilde{G}^t,\tilde{R}^t,\chi^t)$ has the same number of variables as equality constraints, we can delete an equality constraint and obtain an equivalent linear program with more variables than equality constraints.
Thus, the extreme point $x^{t+1}$ either has some integral entry or satisfies an inequality constraint with equality.
In the former case, this reduces the number of variables in the iteration $t+1$; in the latter case, it reduces the number of inequality constraints.
Since an integral entry of $x^t$ never becomes fractional at $x^{t'}$ with $t'>t$ and an equality constraint at $t$ never becomes an inequality constraint at $t'>t$, we conclude that the algorithm terminates. 
Moreover, in each iteration the algorithm performs operations that take time linear in $|A|,|R|,\sum_{\ell\in [d]}k_\ell$, and $|\calF^t|$, while the linear program has $|\calF^t|$ variables and $\calO\big(|A|+|R|+\sum_{\ell\in [d]}k_\ell + |\calF^t|\big)$ constraints.
Since $|\calF^t| \leq |A| \cdot |\{R' \subseteq R: |R'|\leq \omega^*\}| = \calO(|A| \cdot |R|^{\omega^*})$,
we conclude that the algorithm runs in time polynomial in $|A|,k_1,\ldots,k_d$, and $|R|^{\omega^*}$.

In the remainder of the proof, we show that $y$ satisfies the properties claimed in the statement; i.e., that it is a rounding of $x$ and that it is an $(\alpha,\delta,\Delta)$-approximation of $x$ with respect to $u$.
That it is a rounding of $x$ is straightforward, since whenever $x^t(e)\in \{0,1\}$ for some $e\in \calE$ and some $t\in [T]_0$, the algorithm fixes $x^{t+1}(e)=x^t(e)$, thus $y(e)=x^t$ holds as well.
To prove that $y$ is an $(\alpha,\delta,\Delta)$-approximation of $x$ with respect to $u$ is more demanding, as we need to verify that $y$ satisfies constraints \eqref{cons:agents-eq}-\eqref{cons:agents-ub} and \eqref{ineq:groups-app}-\eqref{ineq:resources-total-app}.

To see that $y$ satisfies the set of constraints \eqref{cons:agents-eq}, let $a\in A'$ and note that, since $x$ is a fractional resource allocation, it satisfies the set of constraints \eqref{cons:agents-eq}.
Thus, for every $t\in [T-1]_0$ such that $a\in A(\calF^t)$ we have that $a\in \tilde{A}^t$, so \cref{cons:bind-agents-no-dev} ensures $\textstyle \sum_{q\in \calT_a} x^{t+1}(a,q) = \sum_{e\in \calF^t(a)} x^{t+1}(e)=1.$
For every $t\in [T-1]_0$ such that $a\notin A(\calF^t)$, we have that $x^{t+1}(a,q)=x^t(a,q)$ for every $q\in \calT_a$, so the same equality holds.
When the algorithm terminates, we have $\tilde{A}^t=\emptyset$ and thus $a\notin A(\calF^t)$, which implies the existence of a unique $q\in \calT_a$ such that $x^T(a,q)=1$.
The algorithm thus fixes $x^{T+1}(a,q)=x^T(a,q)$ for every $q\in \calT_a$ and we conclude.

To see that $y$ satisfies the set of constraints \eqref{cons:agents-ub}, we fix $a\in A\setminus A'$.
For every $t\in [T-1]_0$ such that $a\in A(\calF^t)$, either \cref{cons:bind-agents-no-dev} or \cref{cons:agents-no-dev} guarantees that $\textstyle  \sum_{q\in \calT_a} x^{t+1}(a,q) = \sum_{e\in \calF^t(a)} x^{t+1}(e)\leq 1.$
For every $t\in [T-1]_0$ such that $a\notin A(\calF^t)$, we have that $x^{t+1}(a,q)=x^t(a,q)$ for every $q\in \calT_a$, so the same equality holds.
In the final step, we either have $a\in A(\calF^T)$ and at most one $e\in \calF^T(a)$ is rounded up to ensure that $\sum_{e\in \calF^T(a)} x^{T+1}(e)\leq 1$, or $a\notin A(\calF^T)$ and then $x^{T+1}(a,q)=x^T(a,q)$ for every $q\in \calT_a$.

To prove that $y$ satisfies the set of inequalities \eqref{ineq:groups-app}, we fix $\ell\in [d]$ and $i\in [k_\ell]$ arbitrarily and define
$t(\ell,i) = \max\{t\in \{-1,0,1,\ldots,T-1\}: |{\calF^t}(\ell,i)| \geq \alpha_\ell+1\}$
as the latest step at which $G_{\ell,i}$ has $\alpha_\ell+1$ or more associated fractional entries in $\calF^t$; note that we fix $t(\ell,i)=-1$ if this never occurs.
Then, for every $t\in [t(\ell,i)]_0$, it holds that
\begin{align*}
    U_{\ell,i}(x^{t+1}) & =  \sum_{(a,q)\in \calE\setminus \calF^t: a\in G_{\ell,i}} u_a(q) \cdot x^{t+1}(a,q) + \sum_{e\in {\calF^t}(\ell,i)} u(e)\cdot x^{t+1}(e) \\
    & = \sum_{(a,q)\in \calE\setminus \calF^t: a\in G_{\ell,i}} u_a(q)\cdot x^{t}(a,q) + \sum_{e\in {\calF^t}(\ell,i)} u(e)\cdot x^{t}(e) = U_{\ell,i}(x^{t}),
\end{align*}
where the second equality follows from the fact that \Cref{alg:main} fixes $x^{t+1}(e)=x^t(e)$ for every $e\in \calE\setminus \calF^t$ and from \cref{cons:groups-no-dev}.
Thus, $U_{\ell,i}(x^{t(\ell,i)+1}) = U_{\ell,i}(x^0) = U_{\ell,i}(x)$ and we obtain
\begin{align*}
    |U_{\ell,i}(y)-U_{\ell,i}(x)| & = |U_{\ell,i}(x^{T+1})-U_{\ell,i}(x^{t(\ell,i)+1})|\\
    & = \bigg| \sum_{e\in \calF^{t(\ell,i)+1}(\ell,i)} u(e)\cdot \big(x^{T+1}(e)-x^{t(\ell,i)+1}(e)\big)\bigg| \\
    & \leq  \sum_{e\in \calF^{t(\ell,i)+1}(\ell,i)} u(e)\cdot \big|x^{T+1}(e)-x^{t(\ell,i)+1}(e)\big|  < \big|{\calF^{t(\ell,i)+1}}(\ell,i)\big|\cdot U^*_{\ell,i},
\end{align*}
where we used the triangle inequality, that $x^{T+1}(e)\in \{\lfloor x^{t(\ell,i)+1}(e)\rfloor, \lceil x^{t(\ell,i)+1}(e)\rceil\}$ for every $e\in {\calF^{t(\ell,i)+1}}(\ell,i)$, and that $x^{T+1}(e) = x^{t(\ell,i)+1}(e)$ for every other $e$.
From the definition of $t(\ell,i)$ we have $|{\calF^{t(\ell,i)+1}}(\ell,i)| \leq \alpha_\ell$, so we conclude that $|U_{\ell,i}(y)-U_{\ell,i}(x)| < \alpha_\ell\cdot U^*_{\ell,i}$.

We proceed in a similar way to prove that $y$ satisfies the set of inequalities \eqref{ineq:resources-app}. 
We fix $r\in R$ and define
$t(r) = \max\big\{t\in \{-1,0,1,\ldots,T-1\}: \sum_{(a,q)\in {\calF^t}(r)} q(r) \geq \delta+1\big\}$
as the latest step at which $r$ has $\delta+2$ or more associated fractional entries in $\calF^t$; note that we fix $t(r)=-1$ if this never occurs.
Then, for every $t\in [t(r)]_0$,
\allowdisplaybreaks
\begin{align*}
    \sum_{(a,q)\in \calE} q(r)\cdot x^{t+1}(a,q) & = \sum_{(a,q)\in \calE\setminus {\calF^t}(r)} q(r)\cdot x^{t+1}(a,q) + \sum_{(a,q)\in {\calF^t}(r)} q(r)\cdot x^{t+1}(a,q)\\
    & = \sum_{(a,q)\in \calE\setminus {\calF^t}(r)} q(r)\cdot x^{t}(a,q) + \sum_{(a,q)\in {\calF^t}(r)} q(r)\cdot x^{t}(a,q)\\
    & = \sum_{(a,q)\in \calE} q(r)\cdot x^{t}(a,q),
\end{align*}
where the second equality follows from the fact that \Cref{alg:main} fixes $x^{t+1}(e)=x^t(e)$ for every $e\in \calE\setminus \calF^t$ and from \cref{cons:resources-no-dev}.
Thus, $\sum_{(a,q)\in \calE} q(r)\cdot x^{t(r)+1}(a,q) = \sum_{(a,q)\in \calE} q(r)\cdot x^0(a,q) = \sum_{(a,q)\in \calE} q(r)\cdot x(a,q)$, and we obtain
\begin{align*}
    \bigg|\sum_{(a,q)\in \calE} q(r)\cdot (y(a,q) - x(a,q))\bigg| & = \bigg|\sum_{(a,q)\in {\calF^{t(r)+1}}(r)} q(r)\cdot \big(x^{T+1}(a,q) - x^{t(r)+1}(a,q)\big)\bigg|\\
    & \leq \sum_{(a,q)\in {\calF^{t(r)+1}}(r)} q(r)\cdot \big|x^{T+1}(a,q) - x^{t(r)+1}(a,q)\big|\\
    & < \sum_{(a,q)\in {\calF^{t(r)+1}}(r)} q(r),
\end{align*}
where we used the triangle inequality, that $x^{T+1}(e)\in \{\lfloor x^{t(r)+1}(e)\rfloor, \lceil x^{t(r)+1}(e)\rceil\}$ for every $e\in {\calF^{t(r)+1}}(r)$, and that $x^{T+1}(e)=x^{t(r)+1}(e)$ for every other $e$.
From the definition of $t(r)$ we know that $\sum_{(a,q)\in {\calF^{t(r)+1}}(r)} q(r) \leq \delta$, so we conclude that $|\sum_{(a,q)\in \calE} q(r)\cdot (y(a,q) - x(a,q))| < \delta.$

Finally, to show that $y$ satisfies inequality \eqref{ineq:resources-total-app}, we define
$t^* = \min\{ t\in [T]_0: A(\calF^t)\setminus \tilde{A}^t \leq \Delta\}$
as the first step at which there are $\Delta$ or less agents in $A(\calF^t)\setminus \tilde{A}$.
Note that $t^*\leq T$, since this inequality is guaranteed for $T$.
For every $t<t^*$, we know that
\begin{align*}
    \sum_{(a,q)\in \calE} \omega_a \cdot x^{t+1}(a,q) & = \sum_{(a,q)\in \calE\setminus \calF^t} \omega_a \cdot x^{t+1}(a,q) + \sum_{(a,q)\in \calF^t} \omega_a \cdot x^{t+1}(a,q) \\
    & = \sum_{(a,q)\in \calE\setminus \calF^t} \omega_a \cdot x^{t}(a,q) + \sum_{(a,q)\in \calF^t} \omega_a \cdot x^{t}(a,q) = \sum_{(a,q)\in \calE} \omega_a \cdot x^{t}(a,q),
\end{align*}
where the second equality follows from the fact that \Cref{alg:main} fixes $x^{t+1}(e)=x^t(e)$ for every $e\in \calE\setminus \calF^t$ and from \cref{cons:resources-total-no-dev}, since $\chi^t=1$.
Therefore,
\begin{equation}
    \sum_{(a,q)\in \calE} \omega_a  \cdot x^{t^*}(a,q) = \sum_{(a,q)\in \calE} \omega_a \cdot x^{0}(a,q) = \sum_{(a,q)\in \calE} \omega_a \cdot x(a,q).\label{eq:no-dev-before-t*}
\end{equation}

We now show two claims that will imply that (i) agents in $A(\calF^{t^*})\setminus \tilde{A}^{t^*}$ are the only agents whose associated allocation can deviate; and (ii) for each of these agents, the deviation in terms of resources is at most $\omega^*$.
\begin{restatable}{claim}{claimNonFracAgents}\label{claim:non-frac-agents}
    For every $t'\in [T]_0$ and $a\in A \setminus (A(\calF^{t'})\setminus \tilde{A}^{t'})$, it holds $\sum_{q\in \calT_a} \big(y(a,q)-x^{t'}(a,q)\big) = 0.$
\end{restatable}
\begin{proof}
    We fix $a\in A \setminus (A(\calF^{t'})\setminus \tilde{A}^{t'})$ and distinguish two cases.
    If $a\in A\setminus A(\calF^{t'})$, then for every $q\in \calT_a$ and every $t\geq t'$ we know that $x^{t+1}(a,q)=x^t(a,q)$.
    Then, $y(a,q)=x^{t'}(a,q)$ for every $q\in \calT_a$ and the equality in the statement follows.
    
    On the other hand, if $a\in \tilde{A}^{t'}$, let $t(a)=\max\{t\in [T-1]_0: a\in \tilde{A}^t\}$ be the latest step in which $a$ belongs to $\tilde{A}^t$.
    Note that $t(a)\geq t'$ because of our assumption.
    Furthermore, for every $t\in \{t',\ldots,t(a)\}$ we have that $a\in \tilde{A}^{t}$.
    Otherwise, we would have $|\calF^{t'}(a)| \geq 2$, $|\calF^t(a)|\leq 1$ for some $t\in \{t'+1,\ldots,t(a)-1\}$, and $|\calF^{t(a)}(a)|\geq 2$, a contradiction to the fact that $\calF^{t+1}(a)\subseteq \calF^t(a)$ for every $t\in [T-1]_0$.
    Therefore, for every $t\in \{t',\ldots,t(a)\}$ we obtain
    \begin{align*}
        \sum_{q\in \calT_a} x^{t+1}(a,q) & = \sum_{e\in (\{a\}\times \calT_a) \setminus {\calF^t}(a)} x^{t+1}(e) + \sum_{e\in {\calF^t}(a)} x^{t+1}(e)\\
        & = \sum_{e\in (\{a\}\times \calT_a) \setminus {\calF^t}(a)} x^{t}(e) + \sum_{e\in {\calF^t}(a)} x^{t}(e) = \sum_{q\in \calT_a} x^{t}(a,q),
    \end{align*}
    where the second equality follows from the fact that \Cref{alg:main} fixes $x^{t+1}(e)=x^t(e)$ for every $e\in \calE\setminus \calF^t$ and from \cref{cons:bind-agents-no-dev}.
    For $t\in \{t(a)+1,\ldots,T\}$, we know that $a\notin A(\calF^t)$ and thus $x^{t+1}(a,q)=x^t(a,q)$ for every $q\in \calT_a$.
    Combining these two facts, we conclude that
    \[
        \sum_{q\in \calT_a} y(a,q) = \sum_{q\in \calT_a} x^{T+1}(a,q) = \sum_{q\in \calT_a} x^{t(a)+1}(a,q) = \sum_{q\in \calT_a} x^{t'}(a,q).\qedhere
    \] 
\end{proof}

\begin{restatable}{claim}{claimFracAgents}\label{claim:frac-agents}
    For every $a\in A(\calF^{t^*})\setminus \tilde{A}^{t^*}$, it holds $|\sum_{q\in \calT_a} \big(y(a,q)-x^{t^*}(a,q)\big) | < 1.$
\end{restatable}
\begin{proof}
    We fix $a\in A(\calF^{t^*})\setminus \tilde{A}^{t^*}$ and define $t(a)=\max\{t\in [T]_0: a\in A(\calF^{t})\setminus \tilde{A}^{t}\}$. Note that $t(a)\geq t^*$.
    We first consider the case with $t(a)<T$. Since $0<\sum_{q\in \calT_a} x^{t^*}(a,q)<1$ due to the definition of $t^*$ and $\sum_{q\in \calT_a} x^{t(a)+1}(a,q)\leq 1$ due to \cref{cons:agents-no-dev}, we have that
    \[
        \bigg|\sum_{q\in \calT_a} \big(x^{t(a)+1}(a,q)-x^{t^*}(a,q)\big) \bigg| < 1.
    \]
    Since $a\notin A(\calF^{t(a)+1})\setminus \tilde{A}^{t(a)+1}$, we know from \Cref{claim:non-frac-agents} that $\sum_{q\in \calT_a} (y(a,q)-x^{t(a)+1}(a,q))=0$. Combining these two facts, we obtain
    \[
        \bigg|\sum_{q\in \calT_a}\big(y(a,q)-x^{t^*}(a,q)\big) \bigg| < 1,
    \]
    as claimed.

    If $t(a)=T$, we know that $\sum_{q\in \calT_a} x^{t(a)+1}(a,q) \leq 1$ from the definition of $x^{T+1}$ in \Cref{alg:main}.
    Since we know that $0<\sum_{q\in \calT_a} x^{t^*}(a,q)< 1$ from the definition of $t^*$, we conclude once again that
    \begin{align*}
        \bigg|\sum_{q\in \calT_a} \big(y(a,q)-x^{t^*}(a,q)\big) \bigg| & = \bigg|\sum_{q\in \calT_a} \big(x^{t(a)+1}(a,q)-x^{t^*}(a,q)\big) \bigg|<1.\qedhere
    \end{align*}
\end{proof}

We can now directly conclude that $y$ satisfies inequality \eqref{ineq:resources-total-app} using the previous claims, since
\begin{align*}
    \bigg|\sum_{a\in A} \sum_{q\in \calT_a} \omega_a \cdot (y(a,q)-x(a,q)) \bigg| & = \bigg|\sum_{a\in A} \omega_a  \sum_{q\in \calT_a} \big(y(a,q)-x^{t^*}(a,q)\big) \bigg| \\
    & \leq \sum_{a\in A(\calF^{t^*})\setminus \tilde{A}^{t^*}} \omega_a \cdot \bigg| \sum_{q\in \calT_a} \big(y(a,q)-x^{t^*}(a,q)\big) \bigg|\\
    & < \omega^* \cdot |A(\calF^{t^*})\setminus \tilde{A}^{t^*}| \leq \omega^* \cdot \Delta,
\end{align*}
where the equality follows from \eqref{eq:no-dev-before-t*}, the first inequality from \Cref{claim:non-frac-agents} and the triangle inequality, the second one from \Cref{claim:frac-agents} and $\omega_a\leq \omega^*$, and the last one from the definition of $t^*$.
\end{proof}

\section{Applications}\label{sec:applications}

In this section, we show how to use our rounding \Cref{thm:main} to get approximation guarantees for several resource allocation problems that fall in our MCRA setting.
In \Cref{subsec:group-fairness}, we provide novel approximation guarantees for MCRA instances under group fairness constraints by rounding an optimal solution of a convex program used to model fairness across groups.
Then, in \Cref{subsec:couples}, we show how our framework can also be combined with stability requirements to obtain approximation guarantees for near-feasible stable allocation under group fairness constraints.
Finally, in \Cref{subsec:apportionment}, we provide enhanced guarantees for near-feasible allocations in multidimensional political apportionment.

\subsection{Group-fairness in Resource Allocation}\label{subsec:group-fairness}

We consider a fairly general resource allocation setting where every agent $a$ has a certain demand of $\omega_a$ resources and a certain utility for each bundle, and the goal is to assign exactly one bundle to each agent. 
Formally, an instance of the {\it assignment-MCRA} problem is a tuple $\calI=(A,R,E,G,\omega,c,u)$, where $A,R,G,\omega$, and $c$ are structured in the same way as in the MCRA, we have a utility function $u_a\colon \calT_a \to \RR_+$ for each agent $a\in A$,
and for every pair $(a,r)\in E\subseteq A\times R$ we say that $r$ is {\it acceptable} for $a$.
We continue to denote $\calE = \{(a,q): a\in A, q\in \calT_a\}$ and $\omega^*=\max_{a\in A}\omega_a$. 

For a given instance and an agent $a\in A$, we let $\calT_{a,E}$ be the set of $\omega_a$-bundles $q\in \calT_a$ such that $(a,r)\in E$ for every $r\in R$ with $q(r)\geq 1$, i.e., bundles made of acceptable resources for $a$. 
To distinguish the notation from that of the previous section, we write $\calM = \{(a,q): a\in A, q\in \calT_{a,E}\}$ for the set of feasible agent-bundle pairs in this context
A mapping $x$ on $\calM$ is a (fractional) resource allocation for $\calI$ if its natural extension $x'$ on $\calE$, where $x'(a,q)=x(a,q)$ for $(a,q)\in \calM$ and $x'(a,q)=0$ for $(a,q)\in \calE \setminus \calM$, is a (fractional) resource allocation for the instance $(A,A,R,G,\omega,c)$ of MCRA, i.e., $A'=A$ so every agent is allocated exactly one bundle. 
We call an instance of assignment-MCRA \textit{fractionally feasible} if it admits at least one fractional resource allocation.

\paragraph{\bf Group fairness.} Similarly to \citet{SMNS24approximation}, we model the fairness requirements across groups by following an optimization-driven approach; namely, our goal is to find a resource allocation that maximizes the sum of a certain objective function of each group's utility.
Formally, given a non-decreasing concave function $f\colon \RR_+ \to \RR_+$ and an instance $\calI$ of assignment-MCRA, consider the following maximization problem:
\begin{equation}
   \max\Big\{ \textstyle\sum_{\ell\in [d]}\sum_{i\in [k_\ell]} f(U_{\ell,i}(x)):x\text{ is a fractional resource allocation}\Big\}.\tag*{\normalfont{\mbox{[Fair]}}}\label{eq:obj-groups}
\end{equation}
Recall that $U_{\ell,i}(x)$ corresponds to the utility of group $G_{\ell,i}$, and therefore, in \ref{eq:obj-groups}, the goal is to find a fractional resource allocation that maximizes the total utility across the groups.
The program \ref{eq:obj-groups} can capture a broad family of natural fairness notions by setting the appropriate function $f$, like the classic utilitarian objective with $f(z)=z$ and the celebrated {\it proportionality} objective by using $f(z)=\ln(z)$; see, e.g.,~\citet{young2020equity}. 
In general, we say that an optimal solution for \ref{eq:obj-groups} is {\it fair with respect to} $f$ and denote it by $x^f$.
We remark that this optimal fractional allocation can be computed using state-of-the-art routines for convex optimization due to the concavity of the objective function; see, e.g.,~\citet{bubeck2015convex}.

\paragraph{\bf Near-feasible fair allocations.} Given an instance of assignment-MCRA and a non-decreasing concave function $f\colon \RR_+ \to \RR_+$, we say that a mapping $y\colon \calM \to \{0,1\}$ is an $(\alpha,\delta,\Delta^+)$-{\it approximately fair allocation with respect to $f$} if the following holds:
\begin{align}
    \textstyle\sum_{q\in \calT_{a,E}} y(a,q) & = 1 \qquad\qquad \text{ for every } a\in A,\label[eq]{cons:agents-eq-gf}\\
    | U_{\ell,i}(y) - U_{\ell,i}(x^f)| & <  \alpha_\ell\cdot U^*_{\ell,i} \quad\; \text{ for every } \ell\in [d],\ i\in [k_\ell],\label[ineq]{ineq:groups-gf}\\
    \textstyle\sum_{a\in A} \sum_{q\in \calT_{a,E}} q(r)\cdot y(a,q) - c(r) & \leq \delta \qquad\qquad \text{for every } r\in R,\label[ineq]{ineq:resources-gf}\\
   \textstyle \sum_{r\in R} \max\{0,\sum_{a\in A} \sum_{q\in \calT_{a,E}} q(r) \cdot y(a,q) - c(r) \} & \leq \Delta^+,\label[ineq]{ineq:resources-abs-gf}
\end{align}
where $\alpha \in \NN_0^d$, and $\delta,\Delta^+\in \NN_0$. 
In such near-feasible allocation, the utility of each group $G_{\ell,i}$ deviates strictly less than $\alpha_\ell\cdot  U^{*}_{\ell,i}$ from the utilities in the fair with respect to $f$ solution $x^f$, the capacity of each resource is exceeded by at most $\delta$, and the total excess with respect to the resource capacities is at most $\Delta^+$.
We remark that since in the assignment-MCRA every agent is binding, the total deviation from $x^f$ on the number of allocated resources is always equal to zero; instead, we aggregate the excess usage of each resource as a sensible parameter in this setting. 
This notion of approximately fair allocations is closely related to those by \citet{procacciaetal} and \citet{SMNS24approximation}, but more general: 
The former work does not take the individual deviations from resource capacities into account, while the latter imposes the group utilities in the fair fractional allocation as lower bounds only.
Using our rounding \Cref{thm:main} along with a bound on the number of fractional variables of the initial fair allocation, we get the following guarantees for near-feasible allocations. 
\begin{restatable}{theorem}{corGroupFairness}\label{cor:group-fairness}
    Let $\calI$ be a fractionally feasible instance of assignment-MCRA and $f\colon \RR_+ \to \RR_+$ a non-decreasing concave function. Let $\alpha\in \NN^d_0$ and $\delta\in \NN_0$ be such that
    \[
        \sum_{\ell \in [d]}\frac{1}{\alpha_\ell+1} + \frac{\omega^*}{\delta+2} \leq \frac{1}{2},
    \]
    and let $\Delta^+= \min\{(\omega^*-1)|A|+\omega^*|R|+(\omega^*+1)\sum_{\ell\in [d]}k_{\ell}, \delta |R|\}$. Then, there exists an $(\alpha,\delta,\Delta^+)$-approximately fair allocation for $\calI$ with respect to $f$. Furthermore, this allocation can be found in time polynomial in $|A|,|R|^{\omega^*}$, and $\sum_{\ell\in [d]}k_{\ell}$.
\end{restatable}

\begin{proof}
Let $\calI=(A,R,E,G,\omega,c,u)$, $f$, $\alpha$, $\delta$, and $\Delta^+$ be as in the statement.
We let $x^f$ be an optimal solution for \ref{eq:obj-groups} constructed as follows.
We first solve the program \ref{eq:obj-groups} and let $x^*\colon \calM\to [0,1]$ be an arbitrary optimal solution. 
We now let $x^f$ be any extreme point of the polytope containing all fractional resource allocations that guarantee each group as much utility as $x^*$; i.e., the following polytope with variables $x\colon \calM\to [0,1]$:
\begin{align}
        \sum_{q\in \calT_{a,E}} x(a,q) & = 1 \qquad\quad \text{ for every } a\in A,\label[cons]{cons:agents-fairness}\\
	\sum_{a\in A} \sum_{q\in \calT_{a,E}} q(r)\cdot x(a,q) & \leq c(r) \qquad \text{for every } r\in R,\label[cons]{cons:resources-cap-fairness}\\
        U_{\ell,i}(x) & \geq U_{\ell,i}(x^*) \qquad \text{for every } \ell\in [d],\ i\in [k_\ell].\label[cons]{cons:utilities-fairness}
\end{align}
Clearly, the extension of $x^f$ to the domain $\calE$, which has value $x^f(a,q)=0$ whenever $(a,q)\notin \calM$, is a fractional resource allocation for the instance $(A,A,R,G,\omega,c)$ of MCRA.
In addition, since $\sum_{q\in \calT_{a,E}}x^f(a,q)=1$ for every $a$, we know that either $x^f$ has only integer components or $\calA(x^f)\neq \emptyset$.
In the former case we conclude the result immediately for $x^f$, so in what follows we assume that $\calA(x^f)\neq \emptyset$.
We take $\psi=1$ and, since
\[
    \frac{\mathds{1}_{\psi=1}}{2} + \sum_{\ell \in [d]}\frac{1}{\alpha_\ell+1} + \frac{\omega^*}{\delta+2} \leq 1,
\]
we can apply \Cref{thm:main} for this instance and allocation, $\alpha$, $\delta+1$, $\Delta=0$, and our utilities $u_a$ for each agent $a\in A$.
This theorem implies the existence of a rounding $y\colon \calM\to \{0,1\}$ such that
\begin{align}
    \sum_{q\in \calT_{a,E}} y(a,q) & = 1 \qquad\quad\ \text{ for every } a\in A,\label[eq]{eq:y-agents-eq-gf}\\
    | U_{\ell,i}(y) - U_{\ell,i}(x^f)| & <  \alpha_\ell \cdot U^{*}_{\ell,i} \quad \text{for every } \ell\in [d],\ i\in [k_\ell],\label[ineq]{ineq:y-groups-app-gf}\\
    \bigg|\sum_{a\in A} \sum_{q\in \calT_a} q(r)\cdot (y(a,q) - x^f(a,q))\bigg| & < \delta+1 \quad\quad \text{for every } r\in R.\label[ineq]{ineq:y-resources-app-gf}
\end{align}
We claim the result for this mapping $y$.
We need to show that $y$ can be found in time polynomial in $|A|,|R|^{\omega^*}$, and $\sum_{\ell\in [d]}k_{\ell}$, and that it satisfies \cref{cons:agents-eq-gf,ineq:groups-gf,ineq:resources-gf,ineq:resources-abs-gf}.

That $y$ mapping can be found in time polynomial in $|A|,|R|^{\omega^*}$, and $\sum_{\ell\in [d]}k_{\ell}$ follows from \Cref{thm:main} and the fact that, before applying this theorem, we solve a convex program with linear constraints and a linear program, both with a number of variables and constraints that are polynomial in these inputs; see, e.g.,~\citet{korte2011combinatorial,bubeck2015convex}. 
That $y$ satisfies \cref{cons:agents-eq-gf,ineq:groups-gf} follows immediately from \cref{eq:y-agents-eq-gf,ineq:y-groups-app-gf}.
That $y$ satisfies \cref{ineq:resources-gf} follows from \cref{ineq:y-resources-app-gf}, the fact that $\sum_{a\in A} \sum_{q\in \calT_a} q(r)\cdot x^f(a,q) \leq c(r)$ for every $r\in R$ since $x^f$ is a fractional resource allocation for $\calI$, and that $\sum_{a\in A} \sum_{q\in \calT_a} q(r)\cdot y(a,q)$ is an integer value for every $r\in R$.

In what follows, we show that \cref{ineq:resources-abs-gf} holds when $\Delta^+$ is defined as in the statement.
On the one hand, it is clear that
\[
    \sum_{r\in R} \max\bigg\{0,\sum_{a\in A} \sum_{q\in \calT_a} q(r) \cdot x^f(a,q) - c(r) \bigg\} \leq \sum_{r\in R}\delta = \delta|R|,
\]
where the inequality follows from the fact that $x^f$ and $\delta$ satisfy \cref{ineq:resources-gf}.
Thus, if $\Delta^+=\delta|R|$, \cref{ineq:resources-abs-gf} is satisfied.
To check that this is also the case when $\Delta^+=(\omega^*-1)|A|+\omega^*|R|+(\omega^*+1)\sum_{\ell\in [d]}k_{\ell}$, we follow a similar approach as \citet{SMNS24approximation} but for arbitrary $\omega^*$.

We need some additional notation.
Similarly to that used in \Cref{sec:main}, we let $\calF=\{e\in \calM: x^f(e)\in (0,1)\}$ denote the agent-bundle pairs corresponding to fractional entries of $x^f$,
\begin{align*}
    \calF(a) & = \{(a', q) \in \calF : a' = a\} \qquad\text{ for every } a\in A,\\
    \calF(\ell, i) & = \{(a, q) \in \calF : a \in G_{\ell,i}\}\qquad\text{for every } \ell\in[d], i\in [k_\ell],\\
    \calF(r) & = \{(a, q) \in \calF : q(r) \geq 1\}  \quad\ \text{for every } r\in R
\end{align*}
denote the subset of such pairs associated with a certain agent, group, and resource, respectively, and $A(\calF) = \{a \in A : |\calF(a)| \geq 1\}$ $G_{\ell}(\calF) = \{i \in [k_\ell] : |\calF(\ell, i)| \geq 1\}$ for each $\ell\in [d]$, and $R(\calF) = \{r \in R : |\calF(r)| \geq 1\}$ denote the agents, groups, and resources with at least one associated fractional entry of $x^f$.

    We let $\tilde{R}= \{r\in R: \sum_{a\in A} \sum_{q\in \calT_{a,E}} q(r)\cdot x^f(a,q) = c(r)\}$ be the subset of resources whose associated \cref{cons:resources-cap-fairness} is tight at $x^f$ and $\tilde{G}_\ell=\{i\in [k_\ell]: U_{\ell,i}(x^f) =U_{\ell,i}(x^*)\}$ the set of groups in the $\ell$th dimension whose associated \cref{cons:utilities-fairness} is tight at $x^f$.
    Since $x^f$ is an extreme point of the polytope given by \cref{cons:agents-fairness,cons:resources-cap-fairness,cons:utilities-fairness}, we know that
    \begin{equation}
        |A(\calF)|+|\tilde{R}| + \sum_{\ell\in [d]}|\tilde{G}_\ell| = |\calF|.\label[eq]{eq:vars-cons-ep-fairness}
    \end{equation}
    On the other hand, since $\sum_{r\in R} q(r)\leq \omega^*$ for all $q\in \calT_a$ and $a\in A$, we know that
    \[
        \sum_{a\in A(\calF)} |\calF(a)| + \sum_{r\in R(\calF)}\sum_{(a,q)\in \calF(r)} q(r) \leq (\omega^*+1)|\calF| = (\omega^*+1)\bigg( |A(\calF)|+|\tilde{R}| + \sum_{\ell\in [d]}|\tilde{G}_\ell| \bigg),
    \]
    where the second equality follows from \cref{eq:vars-cons-ep-fairness}.
    This implies
    \begin{align}
        & \sum_{r\in \tilde{R}}\bigg(\sum_{(a,q)\in \calF(r)} q(r)-2\bigg) + \sum_{r\in R(\calF)\setminus \tilde{R}}\sum_{(a,q)\in \calF(r)} q(r)\nonumber \\
        &\leq  \sum_{a\in A(\calF)} (|\calF(a)|-2) + \sum_{r\in \tilde{R}}\bigg(\sum_{(a,q)\in \calF(r)} q(r)-2\bigg) + \sum_{r\in R(\calF)\setminus \tilde{R}}\sum_{(a,q)\in \calF(r)} q(r) \nonumber\\
        &\leq  (\omega^*-1)( |A(\calF)|+|\tilde{R}| ) + (\omega^*+1)\sum_{\ell\in [d]}|\tilde{G}_\ell| \nonumber\\
        &\leq  (\omega^*-1)(|A|+|R|)+(\omega^*+1)\sum_{\ell\in [d]}k_\ell,\label[ineq]{ineq:ub-degree-resources-gf}
    \end{align}
    where we used in the first inequality that $|\calF(a)|\geq 2$ for each $a\in A(\calF)$ since $\sum_{q\in \calT_{a,E}}x^f(a,q)=1$.
    We now make use of the following simple claim to conclude.
\begin{claim}\label{claim:dev-resources-gf}
    For every $r\in R(\calF)$, we have
    \[
        \sum_{a\in A} \sum_{q\in \calT_{a,E}} q(r) \cdot y(a,q) - c(r) \leq \begin{cases} \sum_{(a,q)\in \calF(r)} q(r)-1 & \text{if } r\in \tilde{R}\\
        \sum_{(a,q)\in \calF(r)} q(r) & \text{otherwise.}
        \end{cases}
    \]
\end{claim}
\begin{proof}
    Let $r\in \tilde{R}$.
    Since the sum $\sum_{(a,q)\in \calF(r)}q(r)x^f(a,q)$ is an integer by the definition of $\tilde{R}$ and is non-zero due to $r\in R(\calF)$, it must be at least $1$. 
    From the feasibility of $x^f$, we then have
    \[
       \sum_{(a,q)\in\calM\setminus \calF(r)}q(r)\cdot x^f(a,q) \leq c(r) - \sum_{(a,q)\in \calF(r)}q(r) \cdot x^f(a,q) \leq c(r) - 1.
    \]
    Combining this inequality with the fact that $y(a,q)=x^f(a,q)$ for every $(a,q)\in \calM\setminus \calF(r)$ and $y(a,q)\leq 1$ for every $(a,q)\in \calF(r)$, we obtain
    \begin{align*}
        \sum_{a\in A} \sum_{q\in \calT_{a,E}} q(r) \cdot y(a,q) - c(r) & \leq \sum_{(a,q)\in\calM\setminus \calF(r)}q(r)\cdot x^f(a,q) + \sum_{(a,q)\in \calF(r)}q(r) - c(r)\\
        & \leq \sum_{(a,q)\in \calF(r)}q(r) - 1,
    \end{align*}
    which concludes the proof for this case.

    Let now $r\in R(\calF)\setminus \tilde{R}$.
    In this case, the fact that $y(a,q)=x^f(a,q)$ for every $(a,q)\in \calM\setminus \calF(r)$ and $y(a,q)\leq 1$ for every $(a,q)\in \calF(r)$, along with the feasibility of $x^f$, imply
    \begin{align*}
        \sum_{a\in A} \sum_{q\in \calT_{a,E}} q(r) \cdot y(a,q) - c(r) & \leq \sum_{(a,q)\in\calM\setminus \calF(r)}q(r)\cdot x^f(a,q) + \sum_{(a,q)\in \calF(r)}q(r) - c(r)\\
        & \leq \sum_{(a,q)\in \calF(r)}q(r).\qedhere
    \end{align*}
\end{proof}    

We can now conclude the proof by observing that
\begin{align*}
    \sum_{r\in R} \max\bigg\{0,\sum_{a\in A} \sum_{q\in \calT_{a,E}} q(r)\cdot y(a,q) - c(r) \bigg\} & \leq \sum_{r\in \tilde{R}}\bigg(\sum_{(a,q)\in \calF(r)} q(r)-2+1\bigg) + \sum_{r\in R(\calF)\setminus \tilde{R}}\sum_{(a,q)\in \calF(r)} q(r) \\
    & \leq (\omega-1)|A|+\omega^* |R|+(\omega^*+1)\sum_{\ell\in [d]}k_\ell,
\end{align*}
where the first inequality follows from \Cref{claim:dev-resources-gf} and the second one from \cref{ineq:ub-degree-resources-gf}.
Hence, \cref{ineq:resources-abs-gf} holds when $\Delta^+=(\omega^*-1)|A|+\omega^*|R|+(\omega^*+1)\sum_{\ell\in [d]}k_{\ell}$.
\end{proof}

\paragraph{\bf Consequences for proportional fairness.} In what follows, we discuss the consequences of our approximation guarantees in \Cref{cor:group-fairness} for the relevant case of proportional fairness, i.e., when we take $f(z)=\ln(z)$ in \ref{eq:obj-groups}.
We recall that in this case, the optimality conditions guarantee that any optimal fractional resource allocation $x^f$ of \ref{eq:obj-groups} satisfies 
\begin{equation*}
\sum_{\ell\in [d]}\sum_{i\in [k_{\ell}]}\frac{U_{\ell,i}(x)}{U_{\ell,i}(x^f)}\le \sum_{\ell\in [d]}k_{\ell}\text{ for every fractional resource allocation }x,
\end{equation*}
which is the classic proportional fairness notion~\citep{young2020equity}. 
Recently, \citet{procacciaetal} studied the design of school allocation policies with provable proportionality guarantees in the presence of groups on the students' side.
Their setting is captured by our assignment-MCRA framework when the agents are the students, the resources are the schools, there is a single dimension ($d=1$) with $k_1$ many groups, and every student $a$ is assigned to precisely one single school ($\omega_a=1$) among their acceptable schools $\calT_{a,E}$.   
Formally, for $d=1$ and $k_1=k$, we say that $y\colon \calM \to \{0,1\}$ is an $(\alpha,\delta,\Delta^+)$-{\it approximately proportional allocation} if it satisfies \eqref{cons:agents-eq-gf}, \eqref{ineq:resources-gf}, \eqref{ineq:resources-abs-gf}, and $U_{i}(y)\ge U_i(x)/k-\alpha\cdot U^*_i$ for every $i\in [k]$ and every fractional resource allocation $x$, where we have omitted the dimension subindex. 

Thanks to \Cref{cor:group-fairness}, we can trade off the values of $\alpha$ and $\delta$ to accommodate the policy-maker priorities in terms of utility approximation and maximum resource capacity augmentation for each resource (i.e., school capacities).
For instance, when $\omega^*=1$, both maximum deviations can be set to small constants; some pairs in the Pareto frontier defined by \Cref{cor:group-fairness} are $(\alpha,\delta)\in \{(2,4),(3,2),(5,1)\}$.
In practice, slight constant deviations from the capacity of each school constitute a natural goal, and this differentiates our result from previous work in this setting, where the focus was restricted to the deviations from group utilities and the total excess of allocated resources~\citep{procacciaetal,SMNS24approximation}.

While \Cref{cor:group-fairness} has no direct implications for the case where deviations $\alpha=0$ from group utilities or $\delta=0$ from school capacities are sought, it is not hard to see that these deviations can be achieved by simply rounding all fractional entries of our initial fair fractional allocation $x^f$ up or down, respectively.
This yields a non-constant deviation with respect to the other objective, potentially up to the order of $\Delta^+$.
We show that this cannot be avoided:
If we require $\alpha=0$ deviations from proportionality, we need to accept non-constant deviations $\delta$ from the schools' capacities; if we require $\delta=0$ deviation from the schools' capacities, we need to accept non-constant deviations $\alpha$ from proportionality.

\begin{restatable}{proposition}{propTightProp}\label{prop:tight-prop}
    For every $\alpha\in \NN_0^d$ and $\delta, \Delta^+\in \NN_0$, there exist instances $\calI, \calI'$ of assignment-MCRA with $d=1$ such that $\calI$ does not admit a  $(0,\delta,\Delta^+)$-approximately proportional allocation and $\calI'$ does not admit an $(\alpha,0,\Delta^+)$-approximately proportional allocation.
\end{restatable}

For $\calI$, we take an instance where all schools have capacity one and all students belong to a different group and have positive utility for a single school.
Proportionality then implies a large violation of the capacity of this school. 
For $\calI'$, we consider an example that \citet{procacciaetal} used to prove a lower bound, linear in the number of groups, on the smallest possible simultaneous deviation from proportionality and aggregate capacity that an allocation can achieve.
It is based on cycles of students from two alternating groups and schools of alternating quality, so that forbidding capacity violations leaves a group with zero utility and the other group with a large utility.
\begin{proof}[Proof of \Cref{prop:tight-prop}]
    Let $\alpha$, $\delta$, and $\Delta^+$ be as in the statement.
    For $\calI$, let $m,n\in \NN$ be such that $m\geq n$, consider agents $A = \{a_1,\dots, a_n\}$ with $a_j\in G_j$ for $j\in [n]$, and resources $R = \{r_1, \dots, r_m\}$ with capacity $c(r_i) = 1$ for $i \in [m]$.
    We consider $\omega_a=1$ for every $a\in A$ and, for simplicity, replace bundles with single resources when referring to utilities and assignments.
    The utilities are given by
     \[
        u_a(r_i) =
        \begin{cases}
        1 & \text{if } i=1,\\
        0 & \text{otherwise},
        \end{cases}
     \] 
    for each $a\in A$ and $i\in [m]$.
    In order to have a $(0,\delta,\Delta^+)$-approximately proportional allocation $y\colon A\times R\to \{0,1\}$, we need that, for every agent $a\in A$,
    \[
        \sum_{j\in [m]}u_a(r_j)\cdot y(a,r_j) \geq \frac{1}{n}\max\bigg\{ \sum_{j\in [m]}u_a(r_j)\cdot x(a,r_j): x \text{ is a resource allocation} \bigg\} = \frac{1}{n},
    \]
    where the second inequality comes from the fact that $x_a\colon A\times R \to \{0,1\}$ given by $x_a(a,r_1)=1$ and $x_b(b,r_{j(b)})=1$ for every $b\in A\setminus \{a\}$, where $j\colon A\setminus \{a\}\to [m]\setminus \{1\}$ is any injective function, constitutes a resource allocation with $\sum_{j\in [m]} u_a(r_j)\cdot x(a,r_j)=1$.
    We thus conclude that $y(a,r_1)=1$ for every $a\in A$ and thus 
    \[
        \sum_{a\in A} y(a,r_1) - c(r_1) = n-1. 
    \]
    The result then follows by taking $n$ such that $n-1> \delta$.
    
    For $\calI'$, we let $n\in \NN$ be an even number, partition the agents $A=[n]$ into two equally sized groups $G_{1} = \{a_1, a_2, \dots, a_{n/2}\}$ and $G_{2} = \{b_1, b_2, \dots, b_{n/2}\}$, and
    consider resources $R = \{r_1, r_2, \dots, r_n\}$ with capacities $c(r) = 1$ for all $r \in R$.
    We again consider $\omega_a=1$ for every $a\in A$ and replace bundles with single resources when referring to utilities and assignments.
    The utilities are given by
    \[
    u_a(r_i) = 
    \begin{cases} 
    1 & \text{if } i \text{ is odd }, \\
    0 & \text{otherwise},
    \end{cases}
    \]
    for every $a\in A$.
    The feasibility set $E \subseteq A \times R$ is constructed cyclically to enforce the alternation between agents and resources: 
    \[
        E = \bigcup_{i\in [n/2]}\{(a_i,r_{2i-1}),(a_i,r_{2i})\} \cup \bigcup_{i\in [n/2-1]}\{(b_i,r_{2i}),(b_i,r_{2i+1})\},
    \]
    where, in slight abuse of notation, we denote $r_{n+1}=r_1$. 
    In order to have an $(\alpha,0,\Delta^+)$-approximately proportional allocation $y\colon A\times R\to \{0,1\}$, we need that $y\in \{y_1,y_2\}$, where
    \begin{align*}
        y_1(a_i,r_j)& =\begin{cases} 1& \text{ if } j=2i-1,\\
        0& \text{otherwise,}\end{cases} \qquad\qquad y_1(b_i,r_j)=\begin{cases} 1& \text{ if } j=2i,\\
        0& \text{otherwise,}\end{cases}\\
        y_2(a_i,r_j)& =\begin{cases} 1& \text{ if } j=2i,\\
        0& \text{otherwise,}\end{cases} \qquad\qquad y_2(b_i,r_j)=\begin{cases} 1& \text{ if } j=2i+1,\\
        0& \text{otherwise.}\end{cases}
    \end{align*}
    It is easy to see that
    \begin{align*}
        \sum_{j\in [n]}u_a(r_j)\cdot y_1(a,r_j) & = 0 \qquad\text{for every } a\in G_2,\\
        \sum_{j\in [n]}u_a(r_j)\cdot y_2(a,r_j) & = 0 \qquad\text{for every } a\in G_1.
    \end{align*}
    However, the maximum utilities under some resource allocation is ${n}/{2}$.
    Indeed, this is precisely the utility attained at $y_1$ for $G_1$ and at $y_2$ for $G_2$.
    Since the maximum utility for some bundle is $U^*_1=U^*_2=1$ for both groups, we conclude that a deviation of ${n}/{2}$ is required, and the result follows by taking $n$ such that $n>2\alpha$.
\end{proof}

\Cref{prop:tight-prop} implies that, in a sense, our result providing constant (but non-zero) deviations from both group utilities and school capacities is the best we can aim for.
The search for the best-possible constants is, however, a natural direction for future work.

We remark that, in addition to providing more flexibility for the maximum deviations, our framework directly handles agent-dependent utility functions, bundles consisting of more than a single resource, and multiple dimensions for the groups, which arise naturally, e.g., when seeking fairness across overlapping groups. 
Our result opens the door to designing proportional allocation policies under several socio-demographic dimensions, enhancing the policy's fairness guarantees.

\paragraph{\bf A new notion of envy-freeness.}\label{subsubsec:envy-freeness}
In the same context of school redistricting $(d=1,\omega^*=1)$, \citet{procacciaetal} also studied the existence of approximately {\it envy-free} allocations.
In their work, an allocation is $\alpha$-envy-free if for every pair of groups $i_1,i_2\in [k]$ there is no alternative allocation where (i) agents in $i_1$ are allocated a subset of those resources allocated to agents in $i_2$ in the original allocation, and (ii) the utility of group $i_1$ increases by more than $\alpha\cdot U^*_{i_1}$ with respect to the original allocation. 
Note that this notion only makes sense if the demands of all agents are the same and the utilities of agents in the same group are the same, which are modeling assumptions in \citet{procacciaetal}.\footnote{In fact, they assume that the utilities of all agents are the same.} 
However, under this envy-freeness notion, they proved a strong impossibility as there is a family of instances for which the deviation grows linearly in the number of agents.

To get around this impossibility, we introduce a relaxed version of envy-freeness, where the ratio between the total utility of group $G_{\ell,i}$ for its allocation and its total utility for the allocation of the group $G_{\ell,j}$ should not be smaller than the ratio between the sizes of these groups.
Note that when groups have unit size, this is equivalent to the classic notion of envy-freeness in fair division of indivisible goods; e.g.,~\citet{moulin2004fair}. 
We still consider common demands, i.e., $\omega_a=\omega^*$ for every $a\in A$. We thus denote the (common) set of possible bundles by $\calT$ for simplicity.
We say that an instance of assignment-MCRA is {\it group-homogeneous} if $\omega_a=\omega^*$ for every $a\in A$ and, for every dimension $\ell\in [d]$, every $i\in [k_\ell]$, and every $a,b\in G_{\ell,i}$, we have $u_a(q)=u_b(q)$ for every $q\in \calT$. We denote this common (group) utility function by $u_{\ell,i}$. 
For a group-homogeneous instance of assignment-MCRA, we say that $y\colon \calM \to \{0,1\}$ (resp. $[0,1]$) is an $(\alpha,\delta)$-{\it approximately envy-free allocation} (resp. fractional) if it satisfies \eqref{cons:agents-eq-gf}, \eqref{ineq:resources-gf}, and 
\begin{equation}
	\frac{|G_{\ell,i}|}{|G_{\ell,j}|} \sum_{b\in G_{\ell,j}} \sum_{q\in \calT} u_{\ell,i}(q)\cdot y(b,q) - U_{\ell,i}(y) < \alpha_\ell\cdot U^*_{\ell,i}\label[ineq]{ineq:groups-app-ef}
\end{equation}
for every $\ell\in [d]$ and $i,j\in [k_\ell]$, where $\alpha\in \NN^d$ and $\delta,\Delta\in \NN$. 
That is, the maximum deviation from our notion of envy freeness for each group $G_{\ell,i}$ is captured by $\alpha_\ell$ times the maximum utility of this group for a single resource, and the meaning of the maximum deviation $\delta$ from resource capacities is the same as in approximately fair allocations.
We omit the total excess $\Delta^+$ because only the trivial bound $\delta|R|$ on this deviation will remain in this case.
We obtain the following result.

\begin{restatable}{theorem}{propEnvyFreeness}\label{prop:envy-free-homog}
Let $\calI$ be a group-homogeneous and fractionally feasible instance of assignment-MCRA, and let $\alpha\in \NN^d_0$ and $\delta\in \NN_0$ be such that 
\[
    \sum_{\ell\in [d]}\frac{2(k_\ell-1)}{\alpha_\ell+1} + \frac{\omega^*}{\delta+1} \leq \frac{1}{2}.
\]
Then, there exists an $(\alpha,\delta)$-approximately envy-free allocation for $\calI$. 
\end{restatable}

These approximately envy-free allocations can be found by applying an iterative rounding procedure analogous to the one described in \Cref{alg:main}, starting from an envy-free fractional allocation.
The difference with our original rounding algorithm is the number of constraints we must maintain in each step, which becomes quadratic in the number of groups in each dimension due to the envy-freeness constraints.
Because of the similarity with the proof of \Cref{thm:main}, we defer the details to \Cref{app:prop:envy-free-homog}, where we prove the existence of an envy-free fractional allocation and explain how to adapt our rounding algorithm to obtain this result.
We remark that the condition in \Cref{prop:envy-free-homog} holds for $\alpha_\ell\in \Theta(d k_{\ell})$ for each $\ell\in [d]$ and $\delta\in \Theta(1)$, thus breaking the linear dependence in $|A|$ from the more stringent version by \citet{procacciaetal}.

\subsection{Stable Matchings with Couples}\label{subsec:couples}
In this section, unlike the previous one, the non-existence of fair resource allocations does not come from group fairness constraints but from stability; while stable matchings are guaranteed to exist in the basic single-demand setting~\citep{gale1962college}, this is not the case for multi-demand agents.
To illustrate how our framework can also accommodate stability requirements, we show how our approach exploits the fact that, for agents with demand up to $\omega^*=2$, the existence of fractional allocations satisfying stability is guaranteed, and rounding these allocations preserves this property~\citep{nearfeasiblematchingcouples}. 
Using our rounding \Cref{thm:main}, we can directly recover a recent guarantee by \citet{nearfeasiblematchingcouples} for the existence of near-feasible stable allocations,
and we further extend this setting to handle group fairness requirements and allocation stability to guarantee the existence of near-feasible, stable, and fair allocations (\Cref{prop:couples}).

\paragraph{\bf Instances and stability.} An instance of \textit{couples-MCRA} is a tuple $(A,R,E,G,c,u,\succ)$, where $A$, $R$, $E$, $G$, $c$, and $u$ are structured in the same way as in assignment-MCRA; see \Cref{subsec:group-fairness} for the details.
We omit the parameter $\omega$ from the instance description as we fix $\omega^*=2$ and partition the agents accordingly into $A=A_1\dot{\cup} A_2$, where $\omega_a=1$ for every $a\in A_1$ (\textit{single agents}) and $\omega_a=2$ for every $a\in A_2$ (\textit{couples}).
Following the notation introduced in \Cref{subsec:group-fairness}, $\calT_{a,E}$ denotes the acceptable bundles for each agent $a\in A$, and $\calM=\{(a,q):a\in A, q\in \calT_{a,E}\}$ 
is the set of feasible agent-bundle pairs. 
Finally, $\succ$ represents a set of linear orders: $\succ_r$ is a linear order over $\{a\in A: (a,q)\in \calM\}$ representing the preferences of each resource $r\in R$ and $\succ_a$ is a linear order over $\calT_{a,E}$ representing the preferences of each agent $a\in A$.

A mapping $x$ on $\calM$ is a (fractional) resource allocation for $\calI$ if its natural extension $x'$ on $\calE = \{(a,q):a\in A, q\in \calT_a\}$, given by $x'(a,q)=x(a,q)$ for every $(a,q)\in \calM$ and $x'(a,q)=0$ for every $(a,q)\in \calE\setminus \calM$, is a (fractional) resource allocation for the instance $(A,\emptyset,R,G,\omega,c)$ of MCRA, where we have replaced $A'$ by $\emptyset$. 
Note that every instance of couples-MCRA admits a fractional allocation due to the absence of binding agents; in particular, the trivial allocation where no agent is assigned is always feasible.
A resource allocation $x$ on $\calM$ is {\it blocked} when some of the following three situations happen:
\begin{enumerate}
    \item There is a pair $(a,r)$ with $a\in A_1$ and $x(a,q)=1$ such that $a$ prefers $r$ to $q(1)$ and $r$ either has remaining capacity or prefers $a$ over some other agent allocated in $x$. 
    \item There is a pair $(a,r)$ with $a\in A_2$ and $x(a,q)=1$ such that $a$ prefers $r$ over $q$ and $r$ either has remaining capacity for the couple $a$ or prefers $a$ over some other agent(s) allocated in $x$.
    \item There is a couple $a\in A_2$ and two different resources $r,r'$ such that $a$ would prefer to be assigned to the resources $r$ and $r'$, one each, over their current allocation in $x$, and each of the resources has remaining capacity or prefers the corresponding element of the couple $a$ over some other agent allocated in $x$.
\end{enumerate}
A resource allocation $x$ is {\it stable for the capacities $c$} if it is not blocked.

\paragraph{\bf Fractional allocations and fairness.}
We now extend our notion of group fairness and approximately fair allocations from \Cref{subsec:group-fairness} to the case where we also require stability, provided the maximum demand is $\omega^*=2$.
Consider the following linear program introduced by \citet{nearfeasiblematchingcouples} with variables $x\colon \calM \to [0,1]$:
\begin{align}
	\textstyle\sum_{(a,q)\in \calM}q(r)\cdot x(a,q) &\leq c(r) \quad \text{for every } r \in R, &\tag*{\normalfont{\mbox{[LP-Stable]}}}\label[cons]{lp:stability-coup}\\
	\textstyle\sum_{q\in \calT_{a,E}} x(a,q) & \leq 1 \quad\quad \text{ for every } a\in A.\nonumber
\end{align}
\citet{nearfeasiblematchingcouples} showed that the polytope given by \ref{lp:stability-coup} has a dominating extreme point $x\colon \calM\to [0,1]$ and that any rounding of such point is a stable allocation.
Since our use of dominating extreme points is restricted to this property, we refer to \citet{nearfeasiblematchingcouples} for the precise definition in this setting.
In what follows, we exploit this fact to extend our result regarding approximate group fairness and find allocations that are both approximately fair and stable. 
Similarly to \Cref{subsec:group-fairness}, given a non-decreasing concave function $f\colon \RR_+ \to \RR_+$ and an instance $\calI$ of couples-MCRA, we consider the maximization problem
\begin{equation}
   \max\Big\{\textstyle\sum_{\ell\in [d]}\sum_{i\in [k_\ell]} f(U_{\ell,i}(x)):x\text{ is a dominating extreme point for }\ref{lp:stability-coup}\Big\},\tag*{\normalfont{\mbox{[Fair-Stable]}}}\label{eq:stable-obj-groups}
\end{equation}
where we have now restricted our feasible set to dominating extreme points of the polytope given by \ref{lp:stability-coup} so that rounding any such point will give a stable allocation.
We say that an optimal solution for \ref{eq:stable-obj-groups} is \textit{stable and fair with respect to $f$} and denote it by $x^f$. 

In terms of computational efficiency, it is unknown whether the problem of finding a dominating extreme point of a polytope defined by hypergraph constraints such as \ref{lp:stability-coup} can be solved in polynomial time.
This has only been answered in the positive for special cases that do not capture the polytope \ref{lp:stability-coup}, e.g.,~\citet{faenza2023scarf,chandrasekaran2024scarf}, and thus we do not know whether the optimization program \ref{eq:stable-obj-groups} can be solved efficiently.
We remark that this is the case even without the group fairness, i.e., in the setting of \citet{nearfeasiblematchingcouples}.

\paragraph{\bf Near-feasible stable and fair allocations.}
Given an instance of couples-MCRA and a non-decreasing concave function $f\colon \RR_+\to \RR_+$, we say that $y\colon \calM\to \{0,1\}$ is an \textit{$(\alpha,\delta,\Delta)$-approximately fair stable allocation with respect to $f$}, for $\alpha\in \NN^d_0$ and $\delta,\Delta\in \NN$, if the following holds:
\begin{align}
   \textstyle \sum_{q\in \calT_{a,E}} y(a,q) & \leq 1 \qquad\qquad \text{for every } a\in A,\label[ineq]{cons:agents-coup}\\
    | U_{\ell,i}(y) - U_{\ell,i}(x^f)| & <  \alpha_\ell\cdot U^*_{\ell,i} \quad\ \text{ for every } \ell\in [d],\ i\in [k_\ell],\label[ineq]{ineq:groups-coup}\\
    \textstyle\sum_{a\in A} \sum_{q\in \calT_{a,E}} q(r)\cdot y(a,q) - c(r) & \leq \delta \qquad\qquad \text{for every } r\in R,\label[ineq]{ineq:resources-coup}\\
    \textstyle\sum_{a\in A} \sum_{q\in \calT_{a,E}} \omega_a \cdot y(a,q) - \sum_{r\in R}c(r) & \leq 2\Delta,\label[ineq]{ineq:resources-total-coup}
\end{align}
and $y$ is stable for capacities $c'\colon R\to \RR_+$ given by $c'(r) = \sum_{a\in A} \sum_{q\in \calT_{a,E}} q(r)\cdot y(a,q)$, i.e., by the actual number of assigned agents.
We remark that the capacity violation conditions \eqref{ineq:resources-coup} and \eqref{ineq:resources-total-coup} can alternatively be expressed in absolute values with respect to the stable and fair allocation $x^f$, which may be more suitable for cases with lower bounds on the capacities; we stick to our setting with upper bounds for consistency. 
A direct application of \Cref{thm:main} yields the following result. 
\begin{restatable}{theorem}{propCouples}\label{prop:couples}
    Let $\calI$ be a fractionally feasible instance of couples-MCRA and $f\colon \RR_+ \to \RR_+$ a non-decreasing concave function. Let $\alpha\in \NN^d_0$ and $\delta\in \NN_0$ be such that
    \[
     \sum_{\ell \in [d]}\frac{1}{\alpha_\ell+1} + \frac{2}{\delta+2} \leq \frac{1}{2}.
    \]
    Then, there exists an $(\alpha,\delta,2)$-approximately fair stable allocation for $\calI$ with respect to $f$.
\end{restatable}
\begin{proof}
Let $\calI=(A,R,E,G,c,u,\succ)$, $f$, $\alpha$, and $\delta$ be as in the statement.
We let $x^f$ be any optimal solution for \ref{eq:stable-obj-groups}.
Clearly, the extension of $x^f$ to the domain $\calE$, which has value $x^f(a,q)=0$ whenever $(a,q)\notin \calM$, is a fractional resource allocation for the instance $(A,\emptyset,R,G,\omega,c)$ of MCRA.
We take $\psi=1$ and, since
\[
    \frac{\mathds{1}_{\psi=1}}{2} + \sum_{\ell \in [d]}\frac{1}{\alpha_\ell+1} + \frac{2}{\delta+2} \leq 1,
\]
we can apply \Cref{thm:main} for this instance and allocation, $\psi$, $\alpha$, $\delta+1$, $\Delta=2$, and our utilities $u_a$ for each agent $a\in A$, implying the existence of a rounding $y\colon \calM\to \{0,1\}$ such that
\begin{align}
    \sum_{q\in \calT_{a,E}} y(a,q) & \leq 1 \qquad\quad\qquad \text{for every } a\in A,\label[eq]{eq:y-agents-coup}\\
    | U_{\ell,i}(y) - U_{\ell,i}(x^f)| & <  \alpha_\ell \cdot U^{*}_{\ell,i} \qquad\ \text{ for every } \ell\in [d],\ i\in [k_\ell],\label[ineq]{ineq:y-groups-app-coup}\\
    \bigg|\sum_{a\in A} \sum_{q\in \calT_{a,E}} q(r)\cdot \big(y(a,q) - x^f(a,q)\big)\bigg| & < \delta+1 \quad\quad\quad\;\; \text{for every } r\in R,\label[ineq]{ineq:y-resources-app-coup}\\
    \bigg|\sum_{a\in A} \sum_{q\in \calT_{a,E}} \omega_a \cdot \big(y(a,q)-x^f(a,q)\big) \bigg| & < 2\cdot \Delta.\label[ineq]{ineq:y-resources-total-app-coup}
\end{align}
We claim the result for this mapping $y$.
We need to show that $y$ is stable for capacities $c'\colon R\to \RR_+$ given by $c'(r) = \sum_{a\in A} \sum_{q\in \calT_{a,E}} q(r)\cdot y(a,q)$, and that it satisfies \cref{cons:agents-coup,ineq:groups-coup,ineq:resources-coup,ineq:resources-total-coup}.
The former follows from \citet[Lemma 3]{nearfeasiblematchingcouples}, because $x^f$ is a dominating extreme point for \ref{lp:stability-coup} and $y$ is a rounding of $x^f$.

That $y$ satisfies \cref{cons:agents-coup,ineq:groups-coup} follows immediately from \cref{eq:y-agents-coup,ineq:y-groups-app-coup}.
That $y$ satisfies \eqref{ineq:resources-coup} follows from \eqref{ineq:y-resources-app-coup}, the fact that $\sum_{a\in A} \sum_{q\in \calT_{a,E}} q(r)\cdot x^f(a,q) \leq c(r)$ for every $r\in R$ since $x^f$ is feasible for \ref{lp:stability-coup} and thus a fractional resource allocation for $\calI$, and the fact that $\sum_{a\in A} \sum_{q\in \calT_{a,E}} q(r)\cdot y(a,q)$ is an integer value for every $r\in R$.
Finally, that $y$ satisfies \cref{ineq:resources-total-coup} follows from \cref{ineq:y-resources-total-app-coup} and the fact that $\sum_{a\in A} \sum_{q\in \calT_{a,E}} \omega_a \cdot x^f(a,q) \leq \sum_{r\in R}c(r)$ since $x^f$ is feasible for \ref{lp:stability-coup} and thus a fractional resource allocation for $\calI$.
\end{proof}

We remark that in the case without fairness constraints (i.e., $d=0$), we recover the bounds $\delta=\Delta=2$ from \citet{nearfeasiblematchingcouples}.
Our result, however, allows us to incorporate these constraints while keeping the deviations from resource capacities bounded by small constants.
For example, in the case of a one-dimensional partition of the agents, one can guarantee deviations from group fairness of at most $\alpha_1=5$ by increasing the deviations $\delta$ to $4$, while keeping the deviation from the aggregate capacity at $\Delta=2$.

\subsection{Multidimensional Apportionment}\label{subsec:apportionment}

In the {\it multidimensional apportionment problem}, introduced by \citet{balinskidemange1989a,balinskidemange1989b} for the case of two dimensions and extended by \citet{multidimensionalpoliticalappointment} to an arbitrary number of dimensions $d\ge 2$, the goal is to allocate the seats of a representative body proportionally across several dimensions. 
Classic apportionment methods, e.g., divisor methods or Hamilton's method, aim to assign seats to districts proportionally to their population or political parties proportionally to their electoral support~\citep{balinski2010fair,pukelsheim2017}.
However, one may consider several natural dimensions to decide the seat allocation, both by requiring proportionality across them with respect to their electoral support and by incorporating bounds on the number of seats that groups within them should receive.
In addition to political parties and geographical divisions, natural dimensions include demographics of the elected members such as gender or ethnicity; see, e.g., \citet{cembranocorreadiazverdugo2021,mathieu2022,bonet2024explainable}.

Using our rounding \Cref{thm:main}, we improve over the result by \citet{multidimensionalpoliticalappointment} to get enhanced near-feasibility guarantees for multidimensional apportionment.
Remarkably, we can further bound the total deviation from the house size, while the rounding algorithm by \citeauthor{multidimensionalpoliticalappointment} only controls the deviations on each dimension.

\paragraph{\bf Signpost sequences and rounding rules.} 
The core idea behind proportional apportionment, and in particular divisor methods, is to {\it scale and round}. 
Then, to formally introduce the multidimensional apportionment problem, we need to define the idea of a {\it rounding rule}, which in turn requires the definition of a \textit{signpost sequence}.
A signpost sequence is a function $s\colon \NN_0\to \RR_+$ such that $s(0)=0$, $s(t)\in [t-1,t]$ for every $t\in \NN$, and $s(t+1)>s(t)$ for every $t\in \NN$.
Given a signpost sequence $s$, the rounding rule $\llbracket \cdot \rrbracket_s$ is defined as follows: For every $q\in \RR_+$, we let $\llbracket q \rrbracket_s = \{t\}$ if $s(t)<q<s(t+1)$, and $\{t-1,t\}$ if $q=s(t)$.
In simple terms, any value $q\in [t-1,t]$ gets rounded up if $q>s(t)$, down if $q<s(t)$, and we allow both $t-1$ and $t$ as possible roundings is $q=s(t)$. 

\paragraph{\bf Instances and near-feasible apportionments.} An instance of the \textit{multidimensional apportionment} problem, or MA for short, is a tuple $\calI=(G,E,V,b,B,c)$, where $G$ is structured as in MCRA, $E\subseteq \prod_{\ell=1}^{d} G_{\ell,i}$, $V\in \NN^E$, $b_{\ell,i},B_{\ell,i}\in \NN_0$ for every $\ell\in [d]$ and $i\in [k_\ell]$, and $c\in \NN$.
Each $\ell \in [d]$ represents a dimension according to which the candidates are grouped for the election, e.g., their political parties, districts, gender, or ethnicity.
The set $E$ contains the tuples of groups receiving a strictly positive number of votes, which are specified in the tensor $V$.
The values $b_{\ell,i},B_{\ell,i}$ represent a lower and upper bound on the number of seats that should be assigned to group $G_{\ell,i}$, respectively, and $c$ is the size of the house.
We say that a dimension $\ell\in [d]$ is \textit{binding} if $b_{\ell,i}=B_{\ell,i}$ for every $i\in [k_\ell]$, i.e., if there are hard bounds on the number of seats that all groups in this dimension should receive.

In this section, a mapping $x:E\times [c]\to \{0,1\}$ (resp. $[0,1]$) is an {\it apportionment} (resp.\ fractional) if $b_{\ell,i} \leq \sum_{e\in E: e_\ell = i} \sum_{t\in [c]} x(e,t)  \leq B_{\ell,i}$ for every $\ell\in [d]$ and $i\in [k_\ell],$ and $\sum_{e\in E} \sum_{t\in [c]} x(e,t)  = c$, where the first condition ensures the group bounds and the second one guarantees to fulfill the house size.
As usual, we call an instance \textit{fractionally feasible} if it admits a fractional apportionment. 
For an instance of MA, and given $\alpha\in \NN_0^d$ and $\Delta\in \NN_0$, a mapping $y\colon E\times [c] \to \{0,1\}$ is an $(\alpha,\Delta)$-\textit{approximately proportional apportionment} if there exist values $\lambda_{\ell,i}>0$ for every $\ell\in [d]$ and every $i\in [k_{\ell}]$ such that
\begin{align}
	\textstyle\sum_{t\in [c]} x(e,t) & = \textstyle \llbracket V_e \prod_{\ell\in [d]} \lambda_{\ell,e_\ell} \rrbracket_s \quad \text{for every } e\in E,\label[eq]{eq:proportionality}\\
	\textstyle b_{\ell,i} -\alpha_\ell \leq \sum_{e\in E: e_\ell = i} \sum_{t\in [c]} x(e,t) & \leq B_{\ell,i}+\alpha_\ell \quad \text{for every } \ell\in [d],\ i\in [k_\ell],\label[ineq]{ineq:groups-app-app}\\
	\textstyle|\sum_{e\in E} \sum_{t\in [c]} x(e,t) -c | &  \leq \Delta.\label[ineq]{ineq:resources-total-app-app}
\end{align}
In addition to approximately respecting the group bounds and the house size, we have now imposed the natural notion of proportionality in this setting: 
Each group $G_{\ell,i}$ has an associated \textit{multiplier} $\lambda_{\ell,i}$ and each tuple is assigned a number of seats given by its number of votes scaled by all the multipliers associated with groups in the tuple.
The existence of multipliers such that all deviations are zero is guaranteed when $d\in \{1,2\}$ \citep{balinskidemange1989a,balinskidemange1989b,gaffke2008divisor}, but not when $d\geq 3$ \citep{multidimensionalpoliticalappointment}.
However, we exploit the fact that fractional apportionments are guaranteed to exist and that proportionality is kept upon rounding.

\paragraph{\bf Improved guarantee for near-feasible apportionments.}As a starting point for applying our rounding theorem, we use the following linear program, introduced by \citet{multidimensionalpoliticalappointment}:
\begin{align}
	\min \quad\textstyle\sum_{e\in E} \sum_{t\in [c]} x(e,t) \ln({s(t)}/{V_e})& \tag*{\normalfont{\mbox{[LP-MA]}}}\label{eq:lp-obj}\\
	\text{s.t.}\quad\quad \textstyle b_{\ell,i} \leq \sum_{e\in E: e_\ell = i} \sum_{t\in [c]} x(e,t) & \leq B_{\ell,i} \quad \text{for every } \ell\in [d],\ i\in [k_\ell],\nonumber \\
	\textstyle \sum_{e\in E} \sum_{t\in [c]} x(e,t) & = c, \nonumber\\
	x(e,t) & \in [0,1] \quad \text{for every } e\in E,\ t\in [c].\nonumber
\end{align}
Using a primal-dual analysis, it can be shown that \ref{eq:lp-obj} fully characterizes proportional apportionments, in the sense that a mapping $x$ is a proportional apportionment if and only if it is an optimal solution for it~\cite[Theorem 1]{multidimensionalpoliticalappointment}.
While this linear program does not have, in general, an optimal integral solution, we can still round an optimal solution of \ref{eq:lp-obj} using \Cref{thm:main} and maintain the proportionality condition \eqref{eq:proportionality} to obtain improved approximation guarantees for near-feasible apportionments.
The proof consists of a direct application of \Cref{thm:main} for an MCRA instance with all utilities equal to $1$. 
Note that this integer instantiation of the utilities implies the terms $1/(\alpha_\ell+1)$ from the main theorem are replaced by $1/(\alpha_\ell+2)$, which translates into smaller deviations.
\begin{restatable}{theorem}{propMultiApp}\label{prop:multi-apportionment}
	Let $\calI$ be an instance of MA.
	Let $\alpha\in \NN^d_0$ be such that~$\sum_{\ell\in [d]} 1/(\alpha_\ell+2) \leq 1$ and $\Delta\in \NN_0$ defined by
		$\Delta = \min \{\min\{\Delta_\ell:\ell\in [d]\}, \lceil 1/{(1-\sum_{\ell \in [d]} 1/(\alpha_\ell+2))}-2\rceil \}$, where $\Delta_\ell=\alpha_\ell k_\ell$ for every binding dimension $\ell\in [d]$ and $\Delta_\ell=(\alpha_\ell+1)k_\ell-1$ for every non-binding dimension.
	Then, there exists an $(\alpha,\Delta)$-approximate proportional apportionment for $\calI$.
	Furthermore, this solution can be found in time polynomial in $|E|$ and $c$.
\end{restatable}
\begin{proof}
Let $\calI=(G,E,V,b,B,c)$ and $\alpha$ be as in the statement.
We let $A=E\times [c]$ denote a set of agents, where each agent $(e,t)\in A$ will represent the $t$th potential seat allocated to the tuple $e$.
We let $R=\{r\}$ be a unique resource and we consider $\omega_a=1$ for every $a\in A$.
For simplicity, we represent bundles in $\calT=\calT_a$ for any $a\in A$ by a binary value $q\in \{0,1\}$, where $q=1$ represents that $r$ is allocated and $q=0$ that is not.\footnote{That is, in our original notation $q=1$ would represent the bundle $q'\colon R\to \NN$ given by $q'(r)=1$.}
We let $x^*$ be an optimal solution for \ref{eq:lp-obj}.
Clearly, the natural mapping $\tilde{x}^*$ capturing $x^*$ on the domain $\calE$, given by $\tilde{x}^*((e,t),1)=x^*(e,t)$, is a fractional resource allocation for the instance $(A,\emptyset,R,G,\omega,c)$ of MCRA, where we have no binding agents ($A'=\emptyset$) and $\omega_a=1$ for every $a\in A$.
Since there is only one non-empty bundle, we have $\calA(x^*)=\emptyset$, so we can take $\psi=0$ and, since
$\sum_{\ell \in [d]}{1}/{(\alpha_\ell+2)} \leq 1,$
we can apply \Cref{thm:main} for this instance and allocation, utilities $u_a(q)=1$ for every $a\in A$ and $q\in \calT$, $\alpha+1$, and $\lceil 1/{(1-\sum_{\ell \in [d]} 1/(\alpha_\ell+2))}-1\rceil$.
This theorem implies the existence of a rounding $\tilde{y} \colon \calE\to \{0,1\}$ such that, when translated into $y$ on the domain $E\times [c]$ by setting $y(e,t)=y^*((e,t),1)$, it satisfies
\begin{align}
    \bigg| \sum_{e\in E: e_\ell=i}\sum_{t\in [c]}(y(e,t) - x^*(e,t))\bigg| & <  \alpha_\ell +1 \qquad\qquad\qquad\qquad \text{for every } \ell\in [d],\ i\in [k_\ell],\label[ineq]{ineq:y-groups-app-app}\\
    \bigg|\sum_{e\in E}\sum_{t\in [c]} (y(e,t) - x^*(e,t))\bigg| & < \bigg\lceil \frac{1}{1-\sum_{\ell \in [d]} \frac{1}{\alpha_\ell+2}} -1\bigg\rceil.\label[ineq]{ineq:y-resources-app-app}
\end{align}
We claim the result for this mapping $y$.
That this solution can be found in time polynomial in $|E|$ and $c$ follows directly from \Cref{thm:main} and the fact that $x^*$ is found by solving a linear program with $\calO(|E|\cdot c)$ variables and constraints.
It remains to show that $y$ satisfies \cref{eq:proportionality,ineq:groups-app-app,ineq:resources-total-app-app}.
The former follows from \citet[Lemma 5]{multidimensionalpoliticalappointment}, because $x^*$ is an optimal solution for \ref{eq:lp-obj} and $y$ is a rounding of $x^*$.
That $y$ satisfies \cref{ineq:groups-app-app} follows immediately from \cref{ineq:y-groups-app-app}, the fact that $b_{\ell,i}\leq \sum_{e\in E: e_\ell=i}\sum_{t\in [c]} x^*(e,t) \leq B_{\ell,i}$ for every $\ell\in[d]$ and $i\in [k_\ell]$ since $x^*$ is feasible for \ref{eq:lp-obj}, and the fact that $\sum_{e\in E: e_\ell=i}\sum_{t\in [c]} y(e,t)$ is an integer value for every $\ell\in[d]$ and $i\in [k_\ell]$.

Finally, that $y$ satisfies \cref{ineq:resources-total-app-app} when $\Delta=\lceil 1/{(1-\sum_{\ell \in [d]} 1/(\alpha_\ell+2))}-2\rceil$ follows from \cref{ineq:y-resources-app-app}, the fact that $\sum_{e\in E}\sum_{t\in [c]} x^*(e,t) = c$ since $x^*$ is feasible for \ref{eq:lp-obj}, and the fact that $\sum_{e\in E}\sum_{t\in [c]} y(e,t)$ is an integer value.
We next consider the case where $\Delta=\alpha_\ell k_\ell$ for some binding dimension $\ell\in [d]$, i.e., such that $b_\ell=B_\ell$. We fix such $\ell$ and observe that the previous equality implies $\sum_{e\in E: e_\ell=i}\sum_{t\in [c]} x^*(e,t)=b_\ell\in \NN$.
Thus,
\begin{align*}
    \bigg|\sum_{e\in E} \sum_{t\in [c]} y(e,t) -c \bigg| & = \bigg|\sum_{e\in E} \sum_{t\in [c]} (y(e,t) -x^*(e,t)) \bigg|\\
    & = \bigg| \sum_{i\in [k_\ell]} \sum_{e\in E: E_\ell=i} \sum_{t\in [c]} (y(e,t)-x^*(e,t)) \bigg| \\
    & \leq  \sum_{i\in [k_\ell]} \bigg| \sum_{e\in E: E_\ell=i} \sum_{t\in [c]} (y(e,t)-x^*(e,t)) \bigg|  \leq \sum_{i\in [k_\ell]} \alpha_\ell  = \alpha_\ell k_\ell,
\end{align*}
where we used the fact that $\sum_{e\in E}\sum_{t\in [c]} x^*(e,t) = c$ since $x^*$ is feasible for \ref{eq:lp-obj}, the triangle inequality, and \cref{ineq:y-groups-app-app} combined with the fact that both $\sum_{e\in E: e_\ell=i}\sum_{t\in [c]} x^*(e,t)$ and $\sum_{e\in E: e_\ell=i}\sum_{t\in [c]} y(e,t)$ are integer values.
\Cref{ineq:resources-total-app-app} follows directly.
We finally consider the case where $\Delta=(\alpha_\ell+1) k_\ell-1$ for some non-binding dimension $\ell\in [d]$, i.e., such that $b_\ell<B_\ell$. We fix such $\ell$ and observe that
\allowdisplaybreaks
\begin{align*}
    \bigg|\sum_{e\in E} \sum_{t\in [c]} y(e,t) -c \bigg| & = \bigg|\sum_{e\in E} \sum_{t\in [c]} (y(e,t) -x^*(e,t)) \bigg|\\
    & = \bigg| \sum_{i\in [k_\ell]} \sum_{e\in E: E_\ell=i} \sum_{t\in [c]} (y(e,t)-x^*(e,t)) \bigg| \\
    & \leq  \sum_{i\in [k_\ell]} \bigg| \sum_{e\in E: E_\ell=i} \sum_{t\in [c]} (y(e,t)-x^*(e,t)) \bigg|  < \sum_{i\in [k_\ell]} (\alpha_\ell+1) = (\alpha_\ell+1)k_\ell,
\end{align*}
where we used the fact that $\sum_{e\in E}\sum_{t\in [c]} x^*(e,t) = c$ since $x^*$ is feasible for \ref{eq:lp-obj}, the triangle inequality, and \cref{ineq:y-groups-app-app}.
Since the first expression is an integer value, we conclude that it is bounded from above by $(\alpha_\ell+1)k_\ell-1$ and \cref{ineq:resources-total-app-app} holds in this case as well.
\end{proof}

We remark that our guarantees in \Cref{prop:multi-apportionment} improve over the result by \citet{multidimensionalpoliticalappointment} in that we can bound the total deviation $\Delta$ from the house size, while the rounding algorithm by \citeauthor{multidimensionalpoliticalappointment} only allows trading off the deviations on each dimension.
On the one hand, we get small deviations from the house size if the sum $\sum_{\ell \in [d]} 1/(\alpha_\ell+2)$ is not too close to $1$; for example, for~$d=3$, taking $\alpha\in \{(0,6,6),(1,2,4),(2,2,2)\}$ allows $\Delta= \lceil 1/{(1-\sum_{\ell \in [d]} 1/(\alpha_\ell+2))}-2\rceil=2$.

On the other hand, we also get slight deviations from the house size when a dimension $\ell$ has few groups; for example, the case with $k_\ell=2$ (which may arise, for example, if the corresponding dimension is gender) yields deviations from the house size of at most $2\alpha_\ell$ if the dimension is binding and at most $2(\alpha_\ell+1)-1$ otherwise.
Note that, in particular, the existence of a binding dimension $\ell$ with deviations $\alpha_\ell=0$ implies exactly fulfilling the house size.

We finally remark that, in principle, our rounding algorithm could be further applied to more general apportionment settings in which several representative bodies are to be elected, but this would require novel structural results capturing proportionality in this general case.

\section{Discussion and Final Remarks}\label{sec:discussion}

Our work develops a general iterative rounding framework for resource allocation in two‐sided markets that produces near‐feasible allocations while controlling violations in both resource capacities and fairness targets. 
Fed with an appropriate fractional allocation and tuning a small set of deviation parameters, our approach gives a unified and flexible tool for handling fairness in several resource allocation problems, including school allocations, stable matchings, and political apportionment. 
In doing so, it not only recovers guarantees from these specialized settings but also provides new ones for previously omitted objectives and extends prior guarantees to cases involving multi‐demand agents and overlapping, multidimensional group structures. 
The resulting near-feasible solutions provide robust performance guarantees, even when exact feasibility is theoretically impossible or computationally prohibitive. 
Furthermore, the flexibility in choosing the deviation parameters enables policy designers to tailor the trade-offs between efficiency, fairness, and resource augmentation to the needs of a specific application.

Regarding practical implementations of our result, \Cref{thm:main} allows us to bound worst-case deviations and guarantee the existence of integral solutions that meet certain relaxed requirements.
In practice, one could still use the existential side of our result, albeit with an alternative algorithmic approach, aiming to find an integral solution directly with the target maximum deviations.
A possible shortcut for this is to directly solve the integer program obtained by relaxing the original constraints, allowing the target deviations guaranteed by our sufficient condition. 
Another possibility, based on a simple search, is to start by allowing no deviations and then sequentially increase the allowed deviations; our theorem guarantees that a feasible solution will always be obtained before reaching the maximum deviations we establish.
However, these approaches do not guarantee better deviations in the worst case, and they may be computationally inefficient due to the necessary calls of a black-box integer programming solver, as opposed to the linear programming black-box used in our iterative rounding algorithm.

In \Cref{subsec:group-fairness}, we propose a general multidimensional assignment model to handle general group fairness requirements by a convex optimization-driven approach. 
Using our rounding \Cref{thm:main}, we get a range of flexible guarantees on the existence of near-feasible and fair solutions. 
In particular, we introduce a relaxed notion of group envy-freeness that can escape existing impossibilities to accommodate efficient and near-feasible allocations; we believe this new notion might be of independent interest and deserves further study.

The implications in \Cref{subsec:couples} of our main rounding \Cref{thm:main} contribute to a recent line on the computation of near-feasible allocation under stability and complex constraints. 
\citet{OPResearch-stablewithprop} consider a stable matching problem where the hospitals classify their acceptable set of doctors according to types.
While they model proportionality employing lower or upper bounds on the proportion of doctors of each type to be assigned for each hospital, our approach in \Cref{subsec:couples} incorporates general fairness considerations as the objective function in a convex optimization program.
\Cref{prop:couples} ensures controlled deviations that can be tailored to real-world constraints, giving flexibility to policy designers in implementing stability and optimizing trade-offs between efficiency and fairness.
We also believe our framework can be further exploited under different stability concepts, e.g., the group-stability notion considered by \citet{MGMTSC-groupstabilitycomplexconstr} to match families and localities with contracts and budget constraints.

Beyond the application settings showcased in this work, our framework could be further explored in other domains, such as makespan scheduling \citep{newaproxtechnique,feldman2024proportionally} and fair algorithms for clustering \citep{bera2019fair,makarychev2021approximation}.

\newpage
\appendix

\section{Approximately Envy-free Allocations}
\label{app:prop:envy-free-homog}
\propEnvyFreeness*

This theorem follows from two main ingredients: The existence of a fractional allocation satisfying our notion of envy-freeness, and an iterative rounding procedure that starts from this allocation and renders an integral allocation with the claimed deviations. 
This second ingredient follows from an analogous algorithm and proof to the one used in the proof of \Cref{thm:main}, with the only difference that, in the linear program the algorithm solves in each iteration $t$, \cref{cons:resources-total-no-dev} is never imposed (which can be achieved by simply setting $\chi^t=0$ for every $t$) and the group constraints \eqref{cons:groups-no-dev} are replaced by
\begin{align*}
   \sum_{a\in G_{\ell,i}}\sum_{q\in \calT}u_{\ell,i}(q)\cdot (x(a,q)+y(a,q)) & \geq \frac{|G_{\ell,i}|}{|G_{\ell,j}|} \sum_{b\in G_{\ell,j}}\sum_{q\in \calT} u_{\ell,i}(q)\cdot (x(b,q)+y(b,q)),\\
   \sum_{b\in G_{\ell,j}}\sum_{q\in \calT}u_{\ell,j}(q)\cdot (x(b,q)+y(b,q)) & \geq \frac{|G_{\ell,j}|}{|G_{\ell,i}|} \sum_{a\in G_{\ell,i}}\sum_{q\in \calT} u_{\ell,j}(q)\cdot (x(a,q)+y(a,q)),
\end{align*}
for every $i\in \tilde{G}^t_\ell(\calF)$ and every $j\in [k_\ell]\setminus \{i\}$, so that no envy related to group $G_{\ell,i}$ is generated as long as this group has $\alpha_\ell+1$ or more associated fractional variables.
We recall that $x$ are the integer values fixed in previous iterations and $y$ are the variables of the linear program.
This change leads to an increase in the number of these constraints: We still impose one constraint for each agent in $\tilde{A}$ and each resource in $\tilde{R}^t$, but the number of constraints for each group in $\tilde{G}^t_\ell$ now increases from $1$ to $2(k_\ell-1)$. 
The condition to have more variables than equality constraints in iteration $t$ thus becomes
\[
    \sum_{\ell\in [d]}2(k_{\ell}-1) \bigg\lfloor \frac{z}{\alpha_\ell+1} \bigg\rfloor + \bigg\lfloor \frac{\omega^* z}{\delta+1} \bigg\rfloor \leq \bigg\lceil \frac{z}{2}\bigg\rceil \qquad \text{for every }z\in \NN,
\]
which is ensured if we have the condition $\sum_{\ell\in [d]}\frac{2(k_\ell-1)}{\alpha_\ell+1} + \frac{\omega^*}{\delta+1} \leq \frac{1}{2}$
from the statement.

In the remainder of this appendix, we thus show the missing ingredient: The existence of an envy-free fractional allocation to start the iterative rounding with.
This follows from a simple greedy construction.
\begin{claim}
    For every group-homogeneous and fractionally feasible instance $\calI$ of assignment-MCRA, there exists a fractional allocation $x\colon \calM\to [0,1]$ such that, for every $\ell\in [d]$ and $i,j\in [k_\ell]$,
    \[
        U_{\ell,i}(x) \geq \frac{|G_{\ell,i}|}{|G_{\ell,j}|} \sum_{b\in G_{\ell,j}} \sum_{q\in \calT} u_{\ell,i}(q)\cdot x(b,q).
    \]
\end{claim}

\begin{proof}
Let $\calI=(A,R,E,G,\omega,c,u)$ be a group-homogeneous and fractionally feasible instance of assignment-MCRA, and recall that we refer to the common utility function of each group $G_{\ell,i}$ as $u_{\ell,i}$. 
We can construct a fractional allocation satisfying the envy-freeness condition in the statement via a natural greedy procedure, formally described in \Cref{alg:greedy-ef}.
In simple terms, the greedy algorithm starts from the empty allocation and assigns to each agent $a$ at each step $t$, a fraction $\tau$ of this agent's favorite available bundle, where \textit{available} here means that no resource in the bundle has been fully assigned yet.
We update $t$ to $t+\tau$ and continue as long as there exists an eligible agent-bundle pair such that the agent has not received a full allocation and the bundle remains available.
We claim that the continuous limit of this algorithm, i.e., the algorithm run with $\tau\to 0$, produces the desired allocation.
\begin{algorithm}[t]
\caption{Greedy algorithm for fractional envy-free allocations in assignment-MCRA}
\label{alg:greedy-ef}   
\SetAlgoNoLine
\KwIn{Group-homog. and fractionally feasible instance $\calI=(A,R,E,G,\omega,c,u)$ of assignment-MCRA}
\KwOut{fractional resource allocation $x$}
$x(a,q)\gets 0$ for every $(a,q)\in \calM$\;
$A^0\gets A$, $\calT^0\gets \calT$, $t\gets0$\tcp*{unsaturated agents and bundles}
\While{$\calM\cap (A^t\times \calT^t)\neq \emptyset$}{
    \For{$a\in A^t$}{
        $q^t_a \gets \arg\max\{u_a(q): q\in \calT_{a,E}\cap \calT^t\}$ \tcp*{favorite available bundle for $a$}
        $x(a,q^t_a) \gets x(a,q^t_a)+\tau$\;
        \If{$\sum_{q\in \calT_{a,E}}x(a,q)\geq 1$}{
            $A^{t+\tau}\gets A^t\setminus \{a\}$\tcp*{agent $a$ is saturated}
            $T(a)\gets t$
        }
    }
    \For{$q\in \calT^t$}{
        \If{$\sum_{(a,q')\in \calM}q'(r) \geq c(r)$ for some $r\in R$ with $q(r)\geq 1$}{
            $\calT^{t+\tau} \gets \calT^t\setminus \{q\}$\tcp*{bundle $q$ is saturated}
            $T(q)\gets t$
        }
    }
    $t\gets t+\tau$
}
$T(a)\gets t$ for all $a\in A$ for which $T(a)$ is undefined\;
$T(q)\gets t$ for all $q\in \calT$ for which $T(q)$ is undefined\;
{\bf return} $x$
\end{algorithm}

It is not hard to see that the algorithm terminates: For large enough $t$, agents or bundles become saturated.
This is because, as long as an agent is not saturated, its sum of assigned bundles grows. 
Thus, at $t=1$ all agents that remain active must become saturated and the algorithm terminates.

Call $x$ the allocation output by \Cref{alg:greedy-ef} with $\tau\to 0$.
That $x$ is a fractional allocation is not hard to see: 
Because of the updating conditions of the sets of unsaturated agents and bundles, we immediately have $\sum_{q\in \calT_{a,E}}x(a,q)\leq 1$ for every $a\in A$ and $\sum_{(a,q)\in \calM} q(r) x(a,q)\leq c(r)$ for every $r\in R$. 
To see the envy-freeness condition, suppose for the sake of contradiction that there exists $\ell\in [d]$ and $i,j\in [k_\ell]$ such that
\[
    \sum_{a\in G_{\ell,i}} \sum_{q\in \calT} u_{\ell,i}(q)\cdot x(a,q) < \frac{|G_{\ell,i}|}{|G_{\ell,j}|} \sum_{b\in G_{\ell,j}} \sum_{q\in \calT} u_{\ell,i}(q)\cdot x(b,q).
\]
Rearranging and applying the definition of $x$, we obtain
\[
    \frac{1}{|G_{\ell,i}|} \sum_{a\in G_{\ell,i}}\sum_{q\in \calT_{a,E}} u_{\ell,i}(q) \int_{0}^{T(a)}\mathds{1}_{q^t_a=q} \mathrm dt < \frac{1}{|G_{\ell,j}|} \sum_{b\in G_{\ell,j}}\sum_{q\in \calT_{b,E}} u_{\ell,i}(q) \int_{0}^{T(b)}\mathds{1}_{q^t_b=q} \mathrm dt.
\]
This implies that, for some $a\in G_{\ell,i}$, $b\in G_{\ell,j}$ and $t\in [0,T(a)]$, we have $u_{\ell,i}(q^t_a) < u_{\ell,j}(q^t_b)$, which contradicts the definition of $q^t_a$ as the available bundle that maximizes $a$'s utility.
\end{proof}

\bibliographystyle{abbrvnat}
\bibliography{bibliography}

\end{document}